\documentclass[12pt]{article}
\usepackage{amsmath}
\usepackage{graphicx}
\usepackage{enumerate}
\usepackage{natbib}
\usepackage{url} 

\usepackage[misc]{ifsym}
\usepackage{xcolor}
\usepackage{bm}
\usepackage{bbm}
\usepackage{amsfonts}
\usepackage{booktabs}
\usepackage{subcaption}
\usepackage{amsmath}
\usepackage{amsfonts}
\usepackage{hyperref}
\usepackage{comment}
\usepackage{algorithm}
\usepackage{algpseudocode}
\usepackage{threeparttable}
\usepackage{setspace}
\usepackage{amsmath,amssymb,amsthm,mathtools}

\theoremstyle{plain}

\newtheorem{lemma}{Lemma}

\newtheorem{proposition}{Proposition}

\newcommand{\R}{\mathbb{R}}

\newcommand{\E}{\mathbb{E}}

\newcommand{\diag}{\operatorname{diag}}

\newcommand{\blind}{1}

\begin{document}

\def\spacingset#1{\renewcommand{\baselinestretch}%
{#1}\small\normalsize} \spacingset{1}


\if1\blind
{
  \title{\bf Scalable Variable Selection and Model Averaging for Latent Regression Models Using Approximate Variational Bayes}
  \author{Gregor Zens\thanks{
    Corresponding author: \href{mailto:zens@iiasa.ac.at}{zens@iiasa.ac.at}. International Institute for Applied Systems Analysis
(IIASA), Laxenburg, Austria.
    }\hspace{.2cm}\\
    International Institute for Applied Systems Analysis\\
    and \\
    Mark F.J. Steel \\
    Department of Statistics, University of Warwick}
  \maketitle
} \fi

\if0\blind
{
  \bigskip
  \bigskip
  \bigskip
  \begin{center}
    {\LARGE\bf Scalable Variable Selection and Model Averaging for Latent Regression Models Using Approximate Variational Bayes}
\end{center}
  \medskip
} \fi

\bigskip
\begin{abstract}
We propose a fast and theoretically grounded method for Bayesian variable selection and model averaging in latent variable regression models. Our framework addresses three interrelated challenges: (i) intractable marginal likelihoods, (ii) exponentially large model spaces, and (iii) computational costs in large samples. We introduce a novel integrated likelihood approximation based on mean-field variational posterior approximations and establish its asymptotic model selection consistency under broad conditions. To reduce the computational burden, we develop an approximate variational Bayes scheme that fixes the latent regression outcomes for all models at initial estimates obtained under a baseline null model. Despite its simplicity, this approach locally and asymptotically preserves the model-selection behavior of the full variational Bayes approach to first order, at a fraction of the computational cost. Extensive numerical studies - covering probit, tobit, semi-parametric count data models and Poisson log-normal regression - demonstrate accurate inference and large speedups, often reducing runtime from days to hours with comparable accuracy. Applications to real-world data further highlight the practical benefits of the methods for Bayesian inference in large samples and under model uncertainty. 
\end{abstract}

\noindent%
\vfill

\newpage
\spacingset{1.9} 

\section{Introduction}

Latent variable regression models play a central role in modern statistical analysis, particularly when considering non-Gaussian responses such as binary or count data. A canonical latent Gaussian regression framework, for instance, posits a latent response  
\begin{equation}
z_i = \alpha + \bm{x}_i^\top \bm{\beta} + \varepsilon_i, \quad \varepsilon_i \sim \mathcal{N}(0, \sigma^2),\label{eq:zi_regr}
\end{equation}
modeled using covariates $\bm{x}_i$, and with a deterministic or probabilistic mapping from \(z_i\) to an observed response \(y_i\). For instance, probit regression (where \(\sigma^2 = 1\) is typically fixed) uses a deterministic mapping $y_i = \mathbbm{1}(z_i>0)$, whereas Poisson log-normal models $y_i \sim \mathcal{P}(\exp{(z_i)})$ employ a stochastic link, exemplifying the class of univariate link latent Gaussian models (ULLGMs; \citealp{hrafnkelsson2023statistical}; \citealp{steel2024model}). More examples are discussed below. The Bayesian paradigm typically allows for flexible and straightforward inference in such models via data augmentation algorithms (\citealp{tanner1987calculation}). 

In many applications, the selection of relevant covariates to explain the data is central to the statistical exercise at hand. This requires either identifying a single best-performing model based on a selection criterion or adopting a model averaging perspective, where model uncertainty is taken into account by considering an average over all candidate models weighted by their posterior probabilities. In the Bayesian paradigm, both model selection and model averaging are often based on integrated likelihoods of the form
\[
p(\bm{y}) = \int 
p(\bm{y}, \bm{z}, \alpha, \bm{\beta}, \sigma^2) \, d\alpha \, d\bm{\beta} \, d\sigma^2 \, d\bm{z},
\]
where $\bm{y}=(y_1,\dots,y_n)^\top$ and $\bm{z}=(z_1,\dots,z_n)^\top$. 

While adopting $p(\bm{y})$ for accounting for model uncertainty is conceptually appealing, it comes with three major obstacles: First, marginal likelihoods $p(\bm{y})$ are typically unavailable in closed form for non-Gaussian settings, requiring approximate solutions or employing computationally expensive numerical methods, especially when sample sizes are large. Second, the model space usually grows exponentially in the number of covariates. Third, marginal likelihoods (or their approximations) often need to be computed independently for each model. These repeated computations further multiply the computational demand of existing solutions and can quickly become prohibitive. These issues severely limit the appeal of otherwise theoretically sound and rigorous Bayesian model selection and model averaging techniques in settings with large datasets, which become more and more prevalent in various scientific fields.

In this article, we aim to jointly address intractable marginal likelihoods, large model spaces, and the computational burden of repeated integrated likelihood approximations in the context of model uncertainty with large sample sizes in latent variable regression models. First, we consider novel integrated likelihood approximations based on variational posterior approximations and establish their frequentist asymptotic properties for model selection. We show that in many cases, these approximations are simple to compute and achieve model selection consistency for commonly used latent variable regression models. Second, to further reduce computational overhead, we introduce a scalable \textit{approximate variational Bayes} algorithm that avoids iterative updates of latent variables after initial precomputations. We prove that, under mild conditions, this approximate method locally and asymptotically preserves, to first order, the model-selection behavior of the full variational approach, but at a much lower computational cost. Finally, leveraging these approximations, we construct a model space exploration procedure that navigates probabilistically through potentially large sets of candidate models.  

Beyond these methodological developments, we contribute new theoretical and algorithmic insights for (variational) Bayesian inference in specific latent regression settings, including proofs of posterior existence under improper priors, and closed-form variational Bayes (VB) coordinate ascent updates for several illustrative models. We demonstrate the practical performance of our methods through a series of simulations and real-data analyses, showing that the proposed criteria lead to reliable model selection and robust model averaging across a wide range of scenarios.

The remainder of this article is organized as follows. In Sec.~\ref{sec:general}, we introduce a marginal likelihood approximation based on mean-field variational Bayes posterior approximations and discuss related literature. Sec.~\ref{sec:variational_approx} gives details on obtaining mean-field approximations in latent regression models under a set of convenient, improper priors. In Sec.~\ref{sec:avb}, we discuss scalable model selection using a novel approximate variational Bayes algorithm.  Sec.~\ref{sec:illustrations} presents three illustrative examples in models with deterministic links between $z_i$ and $y_i$. Extensions to stochastic links are discussed in Sec.~\ref{sec:nonconjugate}. Sec.~\ref{sec:algorithmic} provides algorithmic details and Sec.~\ref{sec:numerical} presents extensive numerical studies based on artificial data. In Sec.~\ref{sec:applications}, the algorithms are applied to real-world data. Finally, Sec.~\ref{sec:conclusion} discusses  and concludes.

\section{Variational Integrated Likelihood Approximations}
\label{sec:general}

In the general case, consider a $d$-dimensional vector of regression parameters $\bm{\theta}$ and an  $n$-dimensional latent data vector $\bm{z}$ in order to model an $n$-dimensional observed data vector $\bm{y}$. We will focus specifically on cases where the joint density factorizes as $p(\bm{y}, \bm{z}, \bm{\theta}) = p(\bm{y}|\bm{z})p(\bm{z}|\bm{\theta}) p(\bm{\theta})$ where $p(\bm{y}|\bm{z}) = \prod_i p(y_i|z_i)$ is the observed data likelihood given the latent data, $p(\bm{z}|\bm{\theta}) = \prod_i p(z_i|\bm{\theta})$ is the latent (typically Gaussian) data likelihood given the parameters and $p(\bm{\theta})$ is the prior density on $\bm{\theta}$, where $\bm{\theta} = (\alpha, \bm{\beta}, \sigma^2)$ will be the special case discussed in the regression examples below. We assume these quantities are of known form and easy to evaluate. This covers a broad range of widely used latent regression models. Dependence on further observed quantities such as observed regressors $\bm{x}_i$ is not explicit for notational convenience. In this setting, the integrated likelihood $p(\bm{y})$ can be written as
\begin{equation}
\label{eq:main_ml}
p(\bm{y}) = \frac{p(\bm{y}|\bm{z})p(\bm{z}|\bm{\theta})p(\bm{\theta})}{p(\bm{z}, \bm{\theta}| \bm{y})},
\end{equation}
where $p(\bm{z}, \bm{\theta}| \bm{y})$ is the joint posterior of the parameters $\bm{\theta}$ and the latent data $\bm{z}$. Identity (\ref{eq:main_ml}) is known as the \textit{candidate's formula}, the \textit{basic marginal likelihood identity} or \textit{Chib's identity} and has been previously considered in the context of computing marginal likelihoods in MCMC settings (\citealp{besag_89}; \citealp{chib1995marginal}; \citealp{chib2001marginal}). Crucially, exploiting this relationship for integrated likelihood calculations requires evaluation of the posterior density $p(\bm{z}, \bm{\theta}| \bm{y})$, which is typically unavailable. In addition, while the identity holds at any point $(\bm{z}, \bm{\theta})$, it is typically more numerically stable when a high density point is chosen, which is usually unknown as well.

We address both issues by replacing the true posterior distribution with a variational approximation and selecting the evaluation point based on a variational high density point estimate. Specifically, suppose we have a posterior approximation  $p(\bm{z}, \bm{\theta}| \bm{y}) \approx q(\bm{z}, \bm{\theta})$ and posterior point estimates based on posterior expectations $\hat{\bm{z}} = \mathbb{E}_{q(\bm{z},\bm{\theta})}\left[\bm{z}\right]$ and $\hat{\bm{\theta}} = \mathbb{E}_{q(\bm{z},\bm{\theta})}\left[\bm{\theta}\right]$ available. Then,
\begin{equation}
\label{eq:approx1}
p(\bm{y}) \approx \frac{p(\bm{y}|\hat{\bm{z}})p(\hat{\bm{z}}|\hat{\bm{\theta}})p(\hat{\bm{\theta}})}{q(\hat{\bm{z}}, \hat{\bm{\theta}})}
\end{equation}
is an approximation to the integrated likelihood $p(\bm{y})$. Evidently, this approximation becomes exact only when $p(\hat{\bm{z}}, \hat{\bm{\theta}}| \bm{y}) = q(\hat{\bm{z}}, \hat{\bm{\theta}})$. Note that this is a weaker condition than perfect point-wise density approximations $p(\bm{z}, \bm{\theta}| \bm{y}) = q(\bm{z}, \bm{\theta})$. In this article, we specifically focus on the theoretical and empirical behavior of model selection based on simple mean-field variational Bayes posterior approximations of the form $p(\bm{z}, \bm{\theta}| \bm{y}) \approx q(\bm{z})q(\bm{\theta})$,
resulting in a marginal likelihood approximation

\begin{equation}
\label{eq:approx2}
p(\bm{y}) \approx p_{\text{VBC}}(\bm{y}) = \frac{p(\bm{y}|\hat{\bm{z}})\, p(\hat{\bm{z}}|\hat{\bm{\theta}})\, p(\hat{\bm{\theta}})}{q(\hat{\bm{z}})q(\hat{\bm{\theta}})},
\end{equation}
where we will use \textit{VBC} as a shorthand for \textit{variational Bayes criterion}. The VBC approximation typically differs from the evidence lower bound (ELBO), a standard quantity in variational inference. While the ELBO averages over the variational distribution, VBC evaluates at specific posterior expectations $\hat{\bm{z}}$ and $\hat{\bm{\theta}}$. We provide a comparison in the supplementary material (Sec.~\ref{app:elbo_vs_vbc}), showing that VBC deviates from the true marginal likelihood only through point-wise density approximation errors rather than the full Kullback-Leibler divergence. A template for constructing ELBOs in latent regression models is given in Sec.~\ref{sec:additional_details} in the supplementary material. In our empirical exercises considering latent variable regression models, we find that using $\log p_{\text{VBC}}(\bm{y})$ typically leads to slightly improved model selection and averaging performance compared to using the ELBO directly. 

\subsection{Asymptotic Behavior and Model Selection Consistency}

We choose a finite collection of candidate models
\(\mathcal M=\{\mathcal M_{k}:k=0,\dots ,K\}\), parametrized by the $d_k$-dimensional parameter vector $\bm{\theta}_k$. Let $k=k^*$ index the `true' model that generated the data \(\bm y\). Let $\hat{\bm\theta}_k = \mathbb{E}_{q_k(\bm{\theta}_k)}\left[\bm{\theta}_k \right]$ denote the expectation of the variational posterior of $\bm{\theta}_k$ under model $\mathcal{M}_k$, which will be used as variational posterior parameter point estimate. Similarly, let $\hat{\bm z}_k = \mathbb{E}_{q_k(\bm{z}_k)}\left[\bm{z}_k \right]$ denote the expectation of the variational posterior of the latent data vector $\bm{z}_k$ under model $\mathcal{M}_k$. Finally, let  
\begin{equation}
\label{eq:vbc_define}
\mathrm{VBC}_k\;=\;-2
\log\!\Bigl\{
  \frac{p(\bm y\mid \hat{\bm z}_k)\,
        p(\hat{\bm z}_k \mid\hat{\bm\theta}_{k})\,
        p(\hat{\bm\theta}_{k})}
       {q_{k}(\hat{\bm z}_k)\,q_{k}(\hat{\bm\theta}_{k})}
\Bigr\}
\end{equation}
denote the VBC where scaling with $-2$ enhances comparability to classical model selection criteria such as the BIC and where we have made explicit in the notation that the densities $q_k(\bm{z}_k)$ and $q_k(\bm{\theta}_k)$ vary across models $k$.

To assess asymptotic model selection behavior of $\mathrm{VBC}_k$, a set of assumptions is required. Beyond assuming (i) standard regularity conditions\footnote{Standard regularity conditions include the parameter space being compact, the log-likelihood being three times continuously differentiable, the Fisher information matrix being positive definite, and certain moment conditions holding; see, e.g., \citet[Ch.~5]{vandervaart1998asymptotic} for details.}, we assume that, for all $k$ (ii) $q_k(\bm{\theta}_k)$ approaches a multivariate Gaussian density asymptotically in $n$, (iii) this multivariate Gaussian is centered on the `true' model-specific parameter values, denoted as $\bm{\theta}^*_{k}$,\footnote{Conditions (ii) and (iii) can, for instance, be motivated based on related Bernstein-von-Mises results; compare \cite{zhang2024bayesian}.} and (iv) the variational posterior $q_k(\bm{\theta}_k)$ contracts around $\bm{\theta}_k^*$ at the standard parametric rate $1/\sqrt{n}$, implying that the variational posterior variance scales as $1/n$. Under these conditions, we have asymptotic consistency in $\bm{\theta}_k$ in the sense that $q_k(\bm{\theta}_k)$ concentrates at $\bm{\theta}_k^*$ as well as convergence of the variational posterior covariance $V_{\bm{\theta}_k} \xrightarrow[n\rightarrow\infty]{} \frac{1}{n}\bm{\Omega}_{k}$, where $\bm{\Omega}_{k}$ is a (fixed) positive definite matrix and where 
$V_{\bm{\theta}_k} := \mathbb{E}_{q_k(\bm{\theta}_k)}[(\bm{\theta}_k - \hat{\bm{\theta}}_k)(\bm{\theta}_k - \hat{\bm{\theta}}_k)^T]$. Suppressing model subscripts $k$ for simplicity, these conditions lead to the asymptotic relationship 
\begin{equation}
\label{eq:asymptotic_vbc}
\begin{split}
    \log p_{\text{VBC}}(\bm{y}) \xrightarrow[n\rightarrow\infty]{} &\log p(\bm{y}|\hat{\bm{z}}) + \log p(\hat{\bm{z}}|\hat{\bm{\theta}}) + \log p(\hat{\bm{\theta}}) - \log q(\hat{\bm{z}})
 +\\&\frac{d}{2}\log(2\pi) - \frac{d}{2}\log(n)+\frac{1}{2}\log|\bm{\Omega}|.
\end{split}
\end{equation}

This quantity resembles a marginal likelihood approximation as encountered in deriving the Bayesian information criterion (BIC). The asymptotic term (\ref{eq:asymptotic_vbc}) is useful in characterizing model selection consistency based on the VBC in the context of the illustrative model classes discussed in Sec.~\ref{sec:illustrations} and Sec.~\ref{sec:nonconjugate}. Specifically, we can derive the following result:

\begin{proposition}[VBC selects the true model asymptotically]\label{lem:vbc_consistency}

Assume conditions (i) - (iv) as stated above hold. Assume further that, for every model $\mathcal{M}_k$ under consideration, the following condition (v) holds at least asymptotically in $n$:
      \[
      \frac{p(\bm y\mid\hat{\bm z})\,p(\hat{\bm z}\mid\hat{\bm\theta})}
           {q(\hat{\bm z})}\;=\;
      p(\bm y\mid\hat{\bm\theta})
      \]
where the parameters of the variational distribution $q(\bm z)$ were obtained from the variational inference procedure. Then, the VBC rule is asymptotically model selection consistent:
\[
\Pr\!\Bigl\{\arg\min_{k}\mathrm{VBC}_k=k^*\Bigr\}
\;\xrightarrow[n\to\infty]{}1.
\]

\end{proposition}

\begin{proof}[Proof.]
A proof is given in the supplementary material (Sec.~\ref{app:proof_vbc_consistency}).
\end{proof}

Note that condition (v) is key in formally characterizing model selection consistency. It turns out that this property holds for many relevant and widely used latent variable regression models that are based on deterministic relationships of $z_i$ and $y_i$ (discussed in Sec.~\ref{sec:illustrations}). This is due to the combination of the Gaussian  assumption for $z_i$ and $p(y_i|z_i)$ typically restricting the posterior support of $z_i$ to an interval for such models. This results in a ratio of truncated Gaussians and Gaussians, evaluated at the same variational parameters, a scenario where condition (v) holds either exactly or asymptotically; see below for more details and examples. For models  based on stochastic mappings between $z_i$ and $y_i$, condition (v) typically does not hold and we discuss implications of this in Sec.~\ref{sec:nonconjugate}. Nevertheless, we find strong empirical evidence of model selection consistency even if condition (v) is violated, see Sec.~\ref{sec:numerical}. Further important conditions for model selection consistency are asymptotic consistency of the corresponding VB algorithm and posterior existence. For the latent variable regression models considered in Sec.~\ref{sec:illustrations} and Sec.~\ref{sec:nonconjugate}, we will therefore formally characterize posterior existence, and discuss posterior consistency of the respective mean-field VB algorithms in Sec.~\ref{sec:numerical}  and Sec.~\ref{sec:conclusion}.

\subsection{Related Literature}

Previous approaches that use the \textit{basic marginal likelihood identity} to estimate the marginal likelihood include the method of \citet{chib1995marginal}, which applies if all full conditionals have known forms, and the extension of \citet{chib2001marginal}, which allows for Metropolis-Hastings updates but becomes computationally expensive in high dimensions. Modifications of these methods are considered in \cite{de2008improved}. \citet{hsiao2004bayesian} consider multivariate kernel density estimation of the posterior in combination with (\ref{eq:main_ml}), limiting the approach to small-dimensional posteriors. All of these approaches require a sufficient number of MCMC samples to be generated, which is restrictively expensive in the setting we consider. In addition, when many models are to be compared, it is necessary to run the MCMC algorithm for each model, further increasing the computational burden. 

A more efficient alternative is to consider simulation-free approximations to the posterior, as in this article. One option is to employ Laplace approximations, which result in closed-form marginal likelihoods, but do not scale well to high-dimensional posteriors and may impose overly restrictive parametric assumptions. \citet{nott2008approximating} allow for more flexibility and use copulas to approximate the posterior, exploring both simulation-based and simulation-free approaches. \citet{bernardo2003variational} and \citet{mcgrory2007variational} 
directly use the ELBO for model selection. Recent theoretical evidence (\citealp{zhang2024bayesian}) has demonstrated that the model selected by the ELBO indeed tends to asymptotically agree with the one selected by the BIC as the sample size tends to infinity, but that the ELBO tends to incur smaller approximation error to the log marginal likelihood. \citet{kejzlar2023black} discuss a black-box methodology for model averaging using generic variational posterior approximations.

Recently, some `hybrid' approaches have been proposed. \citet{hajargasht2018accurate} use a variational Bayes posterior density estimate as the weighting density in the modified harmonic mean estimator of \citet{geweke1999using}. \citet{chan2024bayesian} use a variational posterior approximation as weighting density in an importance sampling estimation procedure. Both approaches require simulating samples from the variational posterior approximation, which limits the ability to obtain fast approximations, as convergence is expected to be rather slow in high dimensions. Moreover, variational approximations tend to have less variance than the true posterior, which may lead to instabilities in the importance weights.

\section{Obtaining Variational Posterior Approximations}
\label{sec:variational_approx}

\subsection{Model and Prior Setup}

For the latent $z_i$ we consider regression  models as in (\ref{eq:zi_regr}) that are characterized by the inclusion or exclusion of any of the columns of $\bm X$, which is the $n \times p$ matrix with $\bm{x}^\top_{i}$ as its $i$th row. The total number of potential covariates in $\bm X$ is $p$ while $p_k$ indicates the number of covariates from $\bm X$ that are included in model $\mathcal{M}_k$. An intercept term is included in all models.
This results in a model space with $K=2^p$ elements. We consider settings where for model $\mathcal{M}_k$ the distribution of the latent outcomes $\bm{z}=(z_1,\dots,z_n)^\top$ is given by
\begin{equation}
\label{eq:model_k_z}
\bm{z}|\alpha, \bm{\beta}_k, \sigma^2, \mathcal{M}_k \sim \mathcal{N}(\alpha \bm{\iota}_n + \bm{X}_k \bm{\beta}_k , \sigma^2 \bm{I}_n),
\end{equation}
where $\bm{\iota}_n$ is a column vector of $n$ ones, $\bm{I}_n$ is the $n$-dimensional identity matrix, $\bm{X}_{k}$ consists of the $p_k$ columns of $\bm{X}$ that correspond to the regressors that are included in $\mathcal{M}_k$ and $\bm{\beta}_k$ groups the corresponding regression coefficients. The variance term $\sigma^2$ is common to all models and can be fixed or a free parameter to be estimated, depending on the specific modeling choice for $\bm y$ given $\bm z$. The regressors in $\bm{X}$ are centered by subtracting their means, which makes them orthogonal to the intercept and renders the interpretation of the intercept common to all models.

We will focus on a convenient improper prior setup on the parameters that are common to all models, which is often encountered in the context of Bayesian variable selection and model averaging. Specifically, we assume an improper, `non-informative' Jeffreys-type prior on the intercept parameter 
and  $\sigma^2$ 
\begin{equation}
  \label{eq:priora}  p(\alpha) \propto 1,\quad \quad p(\sigma^2) \propto \frac{1}{\sigma^2},
\end{equation}
which 
has the advantage of being invariant with respect to rescaling and translating the $z_{i}$s. For the regression coefficients $\bm{\beta}_k$, we adopt a $g$-prior which is invariant under affine linear transformations of the covariates
\begin{equation}
    \bm{\beta_k} |\sigma^2, \mathcal{M}_k \sim \mathcal{N}( \bm{0}_{p_k}, g\sigma^2(\bm{X}_{k}^\top\bm{X}_{k})^{-1}),\label{PriorThetaMA}
\end{equation}
where $g>0$ is a tuning parameter that we set deterministically to $g=n$. This setup satisfies the conditions required for Proposition~\ref{lem:vbc_consistency} to hold and aligns with choices known to yield consistent model selection in (latent) Gaussian regression settings (\citealp{leysteel2012}; \citealp{steel2024model}). We do not place a prior on $g$ here, but such extensions are natural pathways for future work. Our results are expected to carry over to any point-mass prior on $g$ whose influence is asymptotically dominated by the sample information.

Components of $\bm\beta$ that correspond to excluded regressors under $\mathcal{M}_k$ are assigned a prior point mass at zero for that model. Throughout, we will further assume that the matrix formed by adding a column of ones to $\bm{X}_k$ is of full column rank. If the model space contains models for which this is not the case (for example because $p_k\ge n$), we will assign prior probability zero to those models. For the linear regression model in (\ref{eq:model_k_z}) taken in isolation, this prior setup satisfies many of the desiderata of \cite{Bayarri_etal_12} for objective priors, such as measurement and group invariance and exact predictive matching. However, the considerations of this article extend to different prior setups in a straightforward fashion. 

To derive posterior model probabilities, priors over models have to be specified as well. We consider the beta-binomial structure of \cite{Brown_etal_98}, \cite{ley2009effect} and \cite{ScottBerger}, which amounts to using a Beta$(u,v)$ prior on the common prior inclusion probability for each covariate and results in
\begin{equation}\label{eq:PM_hyper}
P(\mathcal{M}_k)=\frac{\Gamma(u+v)}{\Gamma(u)\Gamma(v)} \frac{\Gamma(u+p_k)\Gamma(v+p-p_k)}{\Gamma(u+v+p)}
\end{equation}
for models where $\bm{X}_k$ is of full column rank. This type of prior is less informative in terms of model size than fixing the prior inclusion probability of the covariates. Following the suggestions of \cite{ley2009effect}, we choose $u=1$ and $v=(p-p_0)/p_0$, where $p_0$ is the prior expected model size over all models for which $\bm{X}_k$ is of full column rank.

\subsection{Variational Posterior}

We consider the class of mean-field approximations where $q(\bm{z}) = \prod_i q_i(z_i)$ and where we assume a suitable factorization exists for \(q(\bm{\theta})\). For latent regression models, we focus on  \(q(\bm{\theta}) = q(\alpha)q(\bm{\beta})q(\sigma^2)\). Throughout, we suppress model subscripts $k$ when working within a single model for notational clarity, reintroducing them only when comparing across models.

The optimal variational densities are obtained using the Coordinate Ascent Variational Inference (CAVI) algorithm, where each variational density is derived by taking the expectation of the joint log density with respect to all other variational densities (\citealp{blei2017variational}), s.t. $q(\bm{\theta}) \propto \exp \left(\mathbb{E}_{-q(\bm{\theta})}\left[\log p(\bm{y}, \bm{z}, \bm{\theta})\right] \right)$, and where $\mathbb{E}_{-q(j)}$ denotes the expectation taken with respect to all variational densities except for $q(j)$. For latent regression models with the prior structure as specified above -- and assuming for now $\sigma^2$ is not fixed -- this CAVI scheme results in closed-form updating rules for $\bm{\theta} = (\alpha, \bm{\beta}, \sigma^2)$. Specifically, the variational densities take the following forms:
\begin{equation}
\begin{aligned}
q(\alpha) &\propto \exp\left( \mathbb{E}_{- q(\alpha)} \left[ \log p(\bm{y}, \bm{z}, \alpha, \bm{\beta}, \sigma^2) \right] \right) \propto \mathcal{N}(\mu_{\alpha}, \omega_{\alpha}), \\
q(\bm{\beta}) &\propto \exp\left( \mathbb{E}_{- q(\bm{\beta})} \left[ \log p(\bm{y}, \bm{z}, \alpha, \bm{\beta}, \sigma^2) \right] \right) \propto \mathcal{N}(\bm{\mu}_{\beta}, \bm{\Omega}_{\beta}), \\
q(\sigma^2) &\propto \exp\left( \mathbb{E}_{- q(\sigma^2)} \left[ \log p(\bm{y}, \bm{z}, \alpha, \bm{\beta}, \sigma^2) \right] \right) \propto \mathcal{IG}(a, b).
\end{aligned}
\end{equation}

 Deriving the variational parameters is relatively straightforward and details can be found in the supplementary material (Sec.~\ref{sec:cavi_updates}). Importantly, up to this point, no specific assumptions have been made regarding the likelihood function $p(y_i | z_i)$ that connects the observed data $y_i$ and the latent data $z_i$, and all the steps outlined are generic. The updating rule for $q(z_i)$ is typically model-dependent and examples are discussed in Secs.~\ref{sec:illustrations} and \ref{sec:nonconjugate}.

\section{Approximate Variational Bayes Approximations}
\label{sec:avb}

Although variational Bayes algorithms are typically more computationally efficient than MCMC methods, they can still pose a significant computational burden when considering many models. In the large $n$ setting, this is mainly because the updating rules for $q(\bm z)$ usually scale at least linearly in the number of observations $n$. Especially when numerical optimization methods for $q(z_i)$ have to be implemented, exploring a large model space becomes computationally infeasible in large sample contexts. Even if efficient model space exploration algorithms are available, the required repeated approximations of $p(\bm{y}|\mathcal{M}_k)$ quickly become prohibitive. The computational burden can be lowered slightly by keeping information on the current $q(z_i)$ for a `warm start' when moving between models, and by parallelizing the updates of each $q(z_i)$ across $i$. However, these updates may remain quite costly, especially when exploring extensive model spaces.

To mitigate this issue, we consider a computationally much cheaper approximation to the variational posterior. Inspired by the work on \textit{approximate Laplace approximations} by \citet{rossell2021approximate}, we consider an \textit{approximate variational Bayes} (AVB) algorithm that scales much more favorably in \( n \) after some precomputations. The key idea is to fix the variational densities \( q(z_i) \) at an initial estimate, which is then not updated anymore and conditioned upon in each iteration and for each model. This makes the algorithm highly scalable for exploring model spaces even for very large samples, since  sufficient statistics such as $\bm{X}^\top\bm{X}$ and $\bm{X}^\top\hat{\bm{z}}$ can be precomputed. 

We specifically consider fixing \( q(z_i) \) based on an initial run of the full variational algorithm under the null model--that is, the model with no covariates present. This implies that the variational densities \( q(z_i) \) are informed solely by the likelihood contributions \( p(y_i \mid z_i) \) as well as the null model mean and variance of the latent variables \( z_i \) across the full sample. More formally, define the null posterior mean at the null intercept $\alpha_0$ and null variance $\sigma^2_0$ by
$m(y_i):=\E[z_i\mid y_i, \alpha_0, \sigma_0^2]$
and the pseudo-outcome $\tilde z_i:=m(y_i)$. Let the stacked vector be $\tilde{\bm{z}}:=(\tilde z_1,\ldots,\tilde z_n)^\top$. The AVB estimator essentially performs a Bayesian regression update using outcome $\tilde{\bm{z}}$ 
and design matrix $(\bm{\iota},\bm{X})$. This setup implies the following local asymptotic relationship between AVB and VB estimates:

\begin{proposition}[Local asymptotic scalar shrinkage of AVB relative to VB]\label{prop:local-shrink-simple}
Assume $p(y_i|z_i)$ is univariate and either an indicator function with respect to a convex interval or a log-concave pdf or pmf that is $C^2$ in $z_i$. Let $\hat{\bm\beta}_{\mathrm{VB},j}$ and $\hat{\bm\beta}_{\mathrm{AVB},j}$ denote the VB and AVB estimators of $\bm\beta$ in a model $\mathcal{M}_j$. Then, in a neighborhood of $\bm{\beta}=0$, as $n\rightarrow\infty$, we have $\hat{\bm\beta}_{\mathrm{AVB},j}
=c\hat{\bm\beta}_{\mathrm{VB},j}
\;+\;o_p\!\bigl(\|\hat{\bm\beta}_{\mathrm{VB},j}\|\bigr)$, with model-independent scaling factor $c\in(0,1)$.
\end{proposition}

\begin{proof}
A proof is given in the supplementary material (Sec.~\ref{app:proof_asymptotic_shrinkage}).
\end{proof}

Intuitively, this result reflects that the bias toward the null model encoded in $q(z_i)$ causes the AVB posterior density of $\bm{\beta}_k$ to be an approximate and asymptotic compromise between $\bm{\beta}=0$ (as implied by the null model) and the full VB point estimate for $\bm{\beta}_k$ (as implied by the usual VB algorithm). An empirical example of this behavior is presented in the supplementary material (Fig.~\ref{fig:avb_convergence}). Note that, as a result, AVB algorithms cannot guarantee consistency for the true parameters in general, even if the corresponding VB algorithm is asymptotically consistent. However, in line with Proposition~\ref{prop:local-shrink-simple}, we find empirically that AVB estimates of $\bm{\beta}_k$ often preserve approximate magnitude ordering and signs of the coefficients. As a consequence, they enable very efficient approximate screening of a large number of models to identify inclusion/exclusion patterns in large sample contexts. More specifically, Proposition~\ref{prop:local-shrink-simple} has the implication that, asymptotically, locally and as a first order approximation, AVB inherits the model selection consistency of VB, which we formalize as follows:

\begin{proposition}[AVB inherits VB model selection consistency locally]
\label{prop:avb-inherits}
Let $\{\mathcal M_k\}_{k=0}^K$ be a finite set of models. Suppose condition (v) holds and, by Proposition~\ref{lem:vbc_consistency}, the VB-based VBC selector is model-selection consistent, i.e., $\Pr\!\left\{\arg\min_{k}\mathrm{VBC}_k^{\mathrm{VB}}=k^*\right\}\to1.$
Assume furthermore that the conditions of  Proposition~\ref{prop:local-shrink-simple} hold, so that uniformly over $k$ in a neighborhood of $\bm\beta=\bm 0$, the AVB plug-in satisfies the local first order asymptotic scalar-shrinkage relation $\hat{\bm\beta}_{\mathrm{AVB},k}
= c\,\hat{\bm\beta}_{\mathrm{VB},k} + o_p\!\big(\|\hat{\bm\beta}_{\mathrm{VB},k}\|\big)$ with $c\in(0,1)$.
Then, still within the above local first order asymptotic regime, the AVB-based VBC selector inherits this model selection consistency:
\[
\Pr\!\left\{\arg\min_{k}\mathrm{VBC}_k^{\mathrm{AVB}}=k^*\right\}\to1.
\]
\end{proposition}
\begin{proof}
    A proof is provided in the supplementary material (Sec.~\ref{app:proof_avb_inherits_consistency}).
\end{proof}

The choice of the null model to fix the variational densities $q(z_i)$ is both a matter of computational convenience (e.g., choosing the full model may not always be possible, for example due to rank-deficient $\bm{X}$) and a conservative choice regarding model selection. Heuristically, if we were to choose any other model composed of a subset of the columns of $\bm{X}$, $q(z_i)$ would be biased towards that particular model. Similarly, all model selection criteria based on these approximate variational approximations would be biased towards that initial model. By selecting the null model, all models are nudged towards the null, biasing the inclusion probabilities downwards. Hence, stronger signals are more likely to `survive' the information loss occurring when moving from VB to AVB. 

Two caveats of this approach are worth mentioning. First, given AVB coefficient posterior means concentrate approximately between 0 and the corresponding VB posterior mean, even if posterior concentration rates and posterior variances were exactly the same for AVB and VB, AVB needs a larger amount of likelihood information than VB to achieve similar asymptotic model selection behavior (compare Fig.~\ref{fig:avb_convergence}). Factors that determine likelihood information and posterior contraction rates in $n$ include the number of covariates $p$ and the correlation structure of the regressors in $\bm{X}$. For latent variable models, there are other factors at play as well, related to the structure of the outcome. For binary outcomes, the balancedness of the outcomes determines likelihood information and affects estimation algorithms (\citealp{zens2024ultimate}). For count outcomes, the likelihood information grows with the size of the counts (\citealp{steel2024model}). We discuss this further in the application section and conduct targeted simulation experiments illustrating this issue in the supplementary Sec.~\ref{app:additional_sim}. Second, our formal results are local, asymptotic and first-order approximate. We thus mainly consider them as a motivation to study AVB empirically rather than as global and exact guarantees. Nonetheless, we find the theoretical predictions to hold more broadly in our extensive empirical exercises. Importantly, in the univariate case ($p=1$), the local and approximate scalar-shrinkage picture of Proposition~\ref{prop:local-shrink-simple} becomes global and exact (Proposition~\ref{prop:avb-shrinkage-1d}; Sec.~\ref{app:avb-shrinkage-1d}).  A full theoretical treatment deriving global and exact results as well as error bounds in the general case is left for future work.

\section{Illustrations for Deterministic Links}
\label{sec:illustrations}
\subsection{Illustration 1: Probit Regression for Binary Data} 

Assume the outcome of interest $\bm{y} \in \{0,1\}^n$ is a $n$-dimensional vector of binary data. In the standard latent variable representation of the probit model (\citealp{albert1993bayesian}), we have $p(y_i | z_i) = \mathbbm{1}(z_i > 0)^{y_i}\mathbbm{1}(z_i \leq 0)^{1-y_i}$. The latent variance term is typically fixed at $\sigma^2 = 1$ a priori as the scale of $z_i$ is unidentifiable given the binary nature of the outcomes $y_i$. Thus, there is no variational factor $q(\sigma^2)$. It is easy to show that the CAVI steps for $q(\alpha)$ and $q(\bm{\beta})$ outlined in Sec.~\ref{sec:cavi_updates} still apply with $a=b=1$ fixed. 

The probit model allows for a closed-form CAVI solution for $q(\bm{z})$, as $\mathbb{E}_{-q(\bm{z})}\left[\log p(\bm{y}, \bm{z}, \alpha, \bm{\beta})\right]$ yields 
\begin{equation}
    q(z_i) = \begin{cases}
    \mathcal{TN}^+(\mu_i, 1),& \text{if } y_i=1\\
    \mathcal{TN}^-(\mu_i, 1),& \text{if } y_i=0,
\end{cases}
\end{equation}
as the appropriate variational update, where $\mathcal{TN}^+$ ($\mathcal{TN}^-$) indicates a normal distribution truncated to the positive (negative) real line and where $\mu_i = \mu_{\alpha} + \bm{x}_i^\top \bm{\mu}_{\beta}$. Additional details including an explicit ELBO term are provided in Sec.~\ref{app:probit_details}. For the probit model, we obtain:

\begin{proposition}
\label{lem:probit_vbc_elbo}
    For the probit model under a multivariate Gaussian prior on $\bm{\beta}$, $p_{\text{VBC}}$ and ELBO perfectly coincide at $(\bm{\hat \theta}, \bm{\hat z})$.
\end{proposition}

\begin{proof}
A proof is given in the supplementary material (Sec.~\ref{app:lemma4}).
\end{proof}

\begin{proposition}
    For the probit model, the posterior of all parameters and the VBC exists under the suggested improper prior setup in (\ref{eq:priora}) (with $\sigma^2=1$) and (\ref{PriorThetaMA}) if not all observed $y_i$'s are equal.
\end{proposition}

\begin{proof}
See the Appendix to \cite{garcia2021bayes}.
\end{proof}

\begin{proposition}
    Model selection based on VBC is asymptotically model selection consistent in the probit case.
\end{proposition}

\begin{proof}
A proof is given in the supplementary material (Sec.~\ref{app:coroll2}).
\end{proof}

\subsection{Illustration 2: Tobit Regression for Censored Data}
\label{sec:tobit}
Consider (without loss of generality) the left-censored Tobit type I regression model, where the relationship between an observed real-valued outcome \(y_i\) and a latent variable \(z_i\) is given by
$y_i = \max(y_L,z_i)$ for some fixed lower bound $y_L$. Thus, when \(y_i>y_L\) the latent variable is observed (i.e. \(z_i=y_i\)), whereas when \(y_i=y_L\) we only know that \(z_i\le y_L\). To make this explicit, define the index sets
$\mathcal{I}_{\mathrm{obs}} = \{ i: y_i>y_L \} \quad \text{and} \quad \mathcal{I}_{\mathrm{lat}} = \{ i: y_i=y_L \}$. For \(i\in\mathcal{I}_{\mathrm{obs}}\) we fix $z_i = y_i$, while for \(i\in\mathcal{I}_{\mathrm{lat}}\) we treat \(z_i\) as latent variable with variational distribution $q(z_i)$. Thus, the joint variational distribution is written as
\[
q(\bm{z},\alpha,\boldsymbol{\beta},\sigma^2) = \left[\prod_{i\in\mathcal{I}_{\mathrm{lat}}} q(z_i)\right]\,
q(\alpha)\,q(\boldsymbol{\beta})\,q(\sigma^2).
\]

Such left-censored data frequently arises in econometric analyses of labor supply or environmental and serological data with measurements falling below detection limits. In this setting, the error variance $\sigma^2$ is a free parameter, but becomes more weakly identified when the share of observations with $y_i = y_L$ increases. The ideas formalized below extend to right-censored data as well as data with upper and lower bounds with minor modifications. A CAVI update for the censored outcome variational densities $q(z_i)$ leads to 
\begin{equation}
    q(z_i) =  \mathcal{TN}(\mu_i, \xi)~~\text{trunc. to}~~(-\infty, y_L]~~\text{if}~~y_i=y_L
\end{equation}
with $\mu_i = \mu_{\alpha} + \bm{x}_i^\top \bm{\mu}_{\beta}$ and $\xi = \frac{b}{a}$. Additional details and a derivation of the ELBO are given in Sec.~\ref{app:tobit_details} in the supplementary material. For the Tobit model, we can derive:
\begin{proposition}
For the left-censored Tobit model, the posterior of all parameters and the VBC exists under the improper prior setup in (\ref{eq:priora}) and (\ref{PriorThetaMA}) if and only if the dataset contains at least two uncensored observations.
\end{proposition}

\begin{proof} 
    A proof is given in the supplementary material (Sec.~\ref{app:lemm5}).
\end{proof}

\begin{proposition}
    Model selection based on the VBC is consistent in the left-censored Tobit case considered.
\end{proposition}

\begin{proof}
A proof is given in the supplementary material (Sec.~\ref{app:coroll3}).
\end{proof}

\subsection{Illustration 3: STAR Models for Count Data}
\label{sec:star}

Assume the outcome of interest $\bm{y} \in \mathbb{N}_0^n$ is an $n$-dimensional vector of counts, i.e., non-negative integers. A semiparametric regression tool for integer-valued data is discussed in \citet{kowal2020simultaneous}. We consider a special case of their framework, based on a simple deterministic transformation and rounding operator applied to a latent Gaussian variable $z_i$, s.t., $y_i = \lfloor \exp(z_i)\rfloor$. Compared to standard parametric count data regression models such as Poisson regression, this so-called Simultaneous Transformation and Rounding (STAR) model has a number of advantages, including closed form conditional posterior densities, and the potential to account for overdispersion.

For $q(\bm{z})$, the CAVI update leads to 
\begin{equation}
    q(z_i) = \begin{cases}
    \mathcal{TN}(\mu_i, \xi)~~\text{trunc. to}~~(-\infty, 0)&  \iff y_i=0\\
    \mathcal{TN}(\mu_i, \xi)~~\text{trunc. to}~~\left[\log(y_i), \log(y_i+1)\right)&  \iff y_i>0\\
\end{cases}
\end{equation}
with $\mu_i = \mu_{\alpha} + \bm{x}_i^\top \bm{\mu}_{\beta}$ and $\xi = \frac{b}{a}$. Additional details and a derivation of the ELBO are given in Sec.~\ref{app:star_details} in the supplementary material. The STAR model leads to the following:
\begin{proposition}
For the considered STAR model, the posterior of all parameters and the VBC exist under the suggested improper prior in (\ref{eq:priora}) and (\ref{PriorThetaMA})  if and only if the data contain at least two positive counts.
\end{proposition}
\begin{proof}
 A proof is given in the supplementary material (Sec.~\ref{app:lemm6}).
\end{proof}

\begin{proposition}
    Model selection based on VBC is asymptotically consistent in the STAR model considered here.
\end{proposition}

\begin{proof}
A proof is given in the supplementary material (Sec.~\ref{app:coroll4}).

\end{proof}

\section{Extensions to Stochastic Links}
\label{sec:nonconjugate}

The three illustrative models considered in Sec.~\ref{sec:illustrations} share certain characteristics, such as closed-form updates for $q(z_i)$ as well as likelihood contributions $p(y_i|z_i)$ that provide information for the posterior of $z_i$ only in the form of interval restrictions. This simplifies certain derivations and proofs of theoretical results and is key for condition (v) in Proposition~\ref{lem:vbc_consistency} to hold. This section considers extensions to cases where no closed-form update for $q(z_i)$ is available and $p(y_i|z_i)$ is a continuous likelihood contribution. Specifically, we consider the class of univariate link latent Gaussian models (\citealp{hrafnkelsson2023statistical}; \citealp{steel2024model}). In this model class, the relationship between $z_i$ and $y_i$ is stochastic and not deterministic. As a specific example, we focus on the Poisson log-normal (PLN) model where $y_i \sim \mathcal{P}(\exp(z_i))$. Many of the insights and derivations below can be extended to other models of this class, including for instance binomial logistic or Erlang log-normal regression models (as defined in \citealp{steel2024model}). Posterior existence for a range of such models using improper priors under mild conditions is discussed in \citet{steel2024model}.

A CAVI update does not lead to a known form for \( q(\bm{z}) \) in the PLN case. Therefore, to infer the variational parameters \( m_i \) and \( s_i \), we will make the parametric assumption that 
$q(z_i)=\mathcal{N}(m_i, s_i)$ and resort to direct optimization of the ELBO in $m_i$ and $s_i$, similar to the approach outlined in \citet{chiquet2018variational}. Details on the ELBO, gradient-based optimization to obtain $m_i$ and $s_i$ as well as some investigation into the quality of approximation of Gaussian $q(z_i)$ for the PLN model are given in Sec.~\ref{app:pln_details}. 

In this section, we focus on examining how our previous results are affected by such settings. Most importantly, for ULLGMs, the key condition (v) $\frac{p(\bm{y}|\bm{z})p(\bm{z}|\bm{\theta})}{q(\bm{z})} = p(\bm{y}|\bm{\theta})$ at the evaluation point $(\hat{\bm{z}}, \hat{\bm{\theta}})$ that established model selection consistency of the VBC in Proposition~\ref{lem:vbc_consistency} does not hold anymore. Instead, we have
\begin{equation}
    \log p(\bm{y}|\bm{\theta}) = \log p(\bm{y}|\bm{z}) + \log p(\bm{z}|\bm{\theta}) - \log q(\bm{z}) + \log \frac{q(\bm{z})}{p(\bm{z}|\bm{y}, \bm{\theta})},
\end{equation}
based on basic probability rules. Hence, asymptotically, $\log p_{\text{VBC}}(\bm{y})$ approaches
\begin{equation}
\label{eq:pln_ms_cons}
    \log p_{\text{VBC}}(\bm{y}) \xrightarrow[n\rightarrow\infty]{} \log p(\bm{y}|\hat{\bm{\theta}})  + \log \frac{q(\hat{\bm{z}})}{p(\hat{\bm{z}}|\bm{y}, \hat{\bm{\theta}})} - \frac{d}{2}\log(n) + \mathcal{O}(1).
\end{equation}

The additional term $\log \frac{q(\hat{\bm{z}})}{p(\hat{\bm{z}}|\bm{y}, \hat{\bm{\theta}})}$ captures the discrepancy between the true posterior of $\bm{z}$ and the corresponding variational approximation at the point $(\hat{\bm{z}}, \hat{\bm{\theta}})$. This error could, in theory, prevent asymptotic model selection consistency. In general, the additional term in (\ref{eq:pln_ms_cons}) may grow linearly in $n$ in the worst case. Thus, for model selection consistency in such cases, we require the additional assumption that, asymptotically, the term $\log \frac{q(\hat{\bm{z}})}{p(\hat{\bm{z}}|\bm{y}, \hat{\bm{\theta}})}$ does not lead to any strong adverse interference with the likelihood term $p(\bm{y}|\hat{\bm{\theta}})$ and the penalty term proportional to $\log(n)$. More specifically, we require that the term cannot diverge to $-\infty$ faster than $\log p(\bm{y}|\hat{\bm{\theta}})$ diverges to $\infty$ as $n$ grows when considering the true model $\mathcal{M}_{k^*}$. Intuitively, this requires that the variational approximation $q(\bm{z})$ is a reasonable approximation to the true conditional $p(\bm{z}|\bm{y}, \bm{\theta})$ at the evaluation point. 

The size of this additional term depends on the model and the data at hand. Typically, any detrimental impact of this term on model selection consistency can be ruled out for various corner cases that do, however, depend on the studied model: For parameter-consistent ULLGMs with Gaussian $q(\bm{z})$, the case where $\sigma^2 \rightarrow 0$ leads to $q(\bm{z})$ being a perfect approximation to $p(\bm{z}|\bm{y}, \bm{\theta})$ and the additional term cancels. As $\sigma^2 \rightarrow \infty$, $p(\bm{z}|\bm{y}, \bm{\theta})$ depends purely on the likelihood contributions $p(\bm{y}|\bm{z})$ which do not vary across models, and hence the additional term cancels when considering Bayes factors. For the PLN model specifically, when $y_i \rightarrow \infty$, $p(z_i|y_i, \bm{\theta})$ approaches a Gaussian, making a Gaussian $q(z_i)$ a perfect approximation and the additional term cancels. Furthermore, asymptotically and given parameter consistency, any model that nests the true model shares the same $\bm{\theta}$ and hence the same $p(\bm{z}|\bm{y}, \bm{\theta})$ and $q(\bm{z})$, leading to the additional term canceling in Bayes factors. These considerations imply that the `bias' induced by the additional term pushes the PLN model at most towards models that are too small. While such considerations can be used to better understand in which situations the additional term may manifest itself, the term itself is difficult to bound formally in the general case. We leave a formal investigation of this issue for future research. Importantly, in virtually all empirical examples studied below, the additional term becomes negligible, and model selection consistency is observed empirically.

\section{Algorithmic Implementation for Large Model Spaces}
\label{sec:algorithmic}

In regression models where the number of covariates $p$ exceeds a certain threshold, enumerating all possible models becomes computationally infeasible. This issue is exacerbated by large sample sizes $n$, since computing model selection criteria for every possible model becomes more costly as $n$ grows. To address this challenge, we consider a simple algorithmic strategy that explores large model spaces in a manner similar to well-known MCMC methods. 

Starting from any model $\mathcal{M}_k$, we compute an approximation $\tilde p_k(\bm{y})$ to the marginal likelihood $p_k(\bm{y})$ and iterate through the following steps: (1) propose a move to a `neighboring' model $\mathcal{M}_j$ using a suitable proposal algorithm; (2) compute a marginal likelihood approximation $\tilde p_j(\bm{y})$; and (3) accept or reject moving to the new model based on the Metropolis-Hastings acceptance probability
\begin{equation}
\alpha = \min\!\Biggl\{ 1,\; 
\frac{\,\tilde p_j(\bm{y})\,p\bigl(\mathcal{M}_j\bigr)\,r\!\bigl(\mathcal{M}_k \mid \mathcal{M}_j\bigr)\,}
     {\,\tilde p_k(\bm{y})\,p\bigl(\mathcal{M}_k\bigr)\,r\!\bigl(\mathcal{M}_j \mid \mathcal{M}_k\bigr)\,}
\Biggr\},
\end{equation}
where $p(\mathcal{M}_k)$ and $p(\mathcal{M}_j)$ are the model priors, and the terms $r(\mathcal{M}_j \mid \mathcal{M}_k)$, $r(\mathcal{M}_k \mid \mathcal{M}_j)$ account for the varying proposal probabilities of moving between models.

We employ a simple add-delete-swap proposal mechanism. Extensions to adaptive proposals (\citealp{zanella2020informed}; \citealp{griffin2021search}) or sampling strategies without replacement (\citealp{clyde2011bayesian}) could enhance efficiency but are left for future research. Various marginal likelihood approximations can be implemented, including ELBO variants, the BIC, or the VBC approximation introduced in Sec.~\ref{sec:general}. The results can be used either to identify a single `best' model based on visit frequency or to derive model-averaged estimates. Note that although the algorithm is inspired by MCMC model space exploration, it essentially defines a quasi-MCMC over models with stationary distribution proportional to the approximate evidence. Hence, its empirical behavior ultimately depends on the quality of the chosen marginal likelihood approximations. We investigate the performance of the algorithm in Sec.~\ref{sec:numerical} and provide full implementation details in the supplementary material (Sec.~\ref{app:algorithmic}).

\section{Numerical Exercises Using Synthetic Data}
\label{sec:numerical}

We simulate synthetic data with various sample sizes between $n=250$ and $n=50,000$ from the respective model-specific data generating process. We consider a setting with \( p = 10 \) where we can enumerate all $1,024$ resulting models, and a setting with \(p = 30\) predictors which requires model space exploration using the algorithm described in Sec.~\ref{sec:algorithmic}. The regressors $\bm{x}_i$ are drawn from a normal distribution with mean 0 and covariance matrix $\bm{\Sigma}$, where $\Sigma_{jk}=\rho^{\vert j - k\vert}$, determined by a correlation coefficient $\rho$. We set $\rho=0.25$ in our examples, representing a setting with moderate correlation among the regressors. The true coefficient vector is specified as $\bm{\beta} = (0.5, -0.5, 0.25, -0.25, 0, \dots, 0)$ in a `sparse' setting and as $\bm{\beta} = (0.5, -0.5, 0.25, -0.25, 0.15, \dots, 0.15)$ in a `dense' setting, and we fix $\alpha=0$ across all settings. All simulation settings are replicated 40 times. For the model space exploration algorithms, we consider inference based on the first $10,000$ models visited after an initial `burn-in' period of $2,000$ models. We consider marginal likelihood approximations based on the ELBO and the VBC, based on AVB and VB plug-in parameter estimates. All marginal likelihood approximations are combined with a beta-binomial model space prior as in \citet{ley2009effect} with a prior mean model size of $p_0=p/2$ and we use a $g$-prior as in (\ref{PriorThetaMA}) with \( g = n \), along with the improper priors in (\ref{eq:priora}).

\textbf{Accuracy.} As a first accuracy measure we consider Brier scores of estimated posterior inclusion probabilities against the true inclusion patterns. All selection criteria considered provide very high quality posterior inclusion probability estimates as $n$ grows in the setting with $p=10$ (Fig.~\ref{fig:brier_enumeration}; log scale in Supplementary Fig.~\ref{fig:brier_enumeration_log}). The results of the full VB and AVB algorithms are essentially indistinguishable, even in moderately large samples of a few thousand observations. In settings with sparse DGP, the bias of the AVB algorithm towards the null model even leads to slightly improved Brier scores relative to the full VB algorithm, which is reversed when the DGP is dense. For the discrete link models, VBC and ELBO give either equivalent or very close results. For the PLN model, VBC performs significantly better in small samples. On average across all settings and models, VB-VBC (Brier: 0.0139) outperforms VB-ELBO (0.0141) slightly. The VB algorithms outperform the AVB algorithms. Among the AVB algorithms, AVB-VBC (0.0227) performs better than AVB-ELBO (0.0312). Importantly, all approximations approach perfect scores as $n$ grows. This provides empirical evidence towards model selection consistency, also in the PLN case where establishing theoretical results is more challenging compared to the other considered examples. As a second measure of model selection performance, the share of simulation runs where the correct model scored best in terms of the respective criterion is provided in Supplementary Fig. \ref{fig:true_enumeration} and \ref{fig:true_enumeration_log}, leading to very similar conclusions. For the setting with $p=30$, where algorithmic exploration of the model space becomes necessary, results are similar, see Supplementary Fig. \ref{fig:brier_exploration}, \ref{fig:brier_exploration_log}, \ref{fig:true_exploration} and \ref{fig:true_exploration_log}.

\begin{figure}[t]
    \centering
    \includegraphics[width=0.95\textwidth]{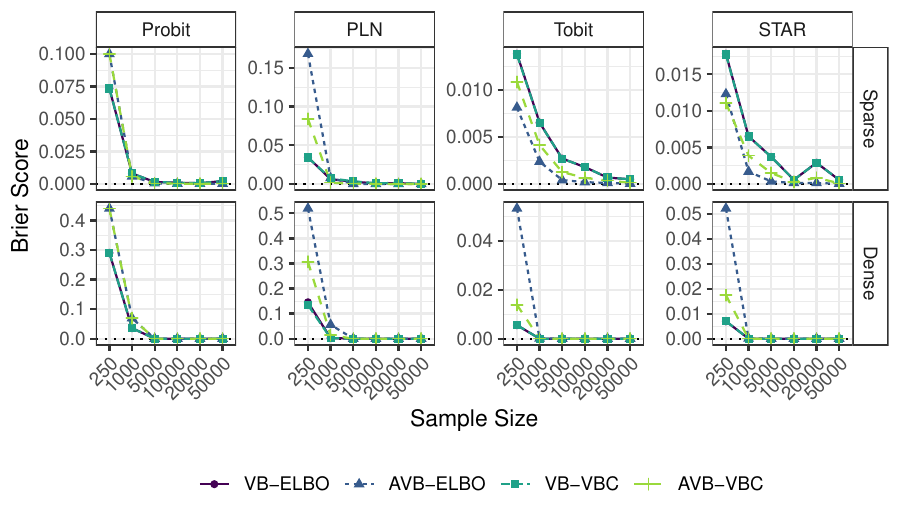}
    \caption{Brier scores in the $p=10$ setting based on model space enumeration using the VBC based on a full VB algorithm (VB-VBC) and approximate variational approximations (AVB-VBC) as well as on the ELBO of the full VB algorithm (VB-ELBO) and an AVB algorithm (AVB-ELBO). Black dotted line refers to zero (perfect accuracy). Results are averaged across 40 replicate datasets.}
    \label{fig:brier_enumeration}
\end{figure}

\textbf{Timing.} Importantly, while AVB-VBC provides very similar accuracy compared to VB-VBC, it comes at a much smaller computational cost (Fig.~\ref{fig:timing_enumeration}; log scale in Supplementary Fig.~\ref{fig:timing_enumeration_log}). 
This is already observable in moderate sample sizes. For instance, in the sparse DGP setting with $n=1,000$, the full probit VB algorithm requires 4.8 seconds to enumerate $1,024$ models, while the AVB algorithm takes 1.4 seconds (a reduction of around 71\%). However, the strongest gains in computational efficiency become visible in large samples. On average across all models and settings, for $n=50,000$ the VB algorithm requires around 742 seconds, while the AVB algorithm takes 14 seconds (a reduction of more than $98\%$). The largest gain we observe is in the $n=50,000$ dense setting of the STAR model, where computation time drops from 1402 seconds to 16 seconds ($-99\%$). The average, enumeration-based Brier scores in this setting are numerically indistinguishable from 0 for both VB and AVB. For the setting with $p=30$, where algorithmic exploration of the model space becomes necessary, results are again similar, see Figs. \ref{fig:timing_exploration} and \ref{fig:timing_exploration_log}.

\begin{figure}[t]
    \centering
    \includegraphics[width=0.95\textwidth]{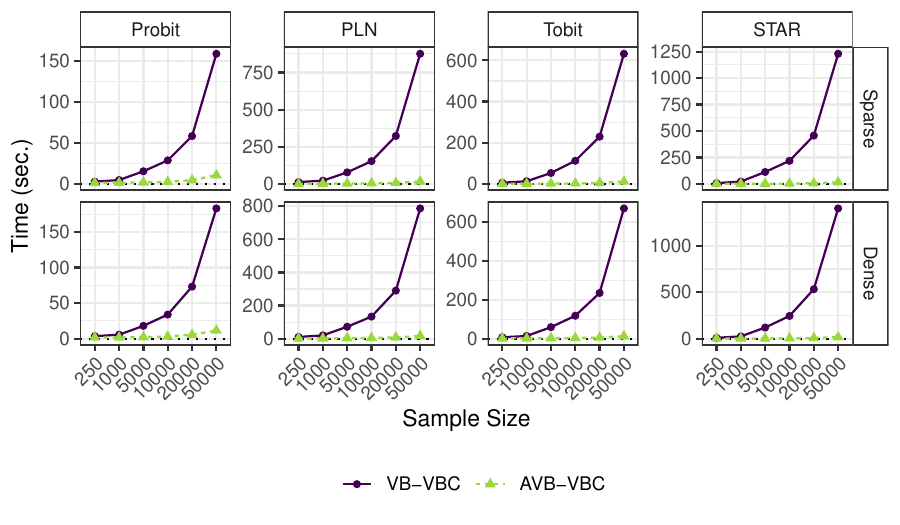}
    \caption{Timing in seconds to evaluate $1,024$ models using the VBC based on a full VB algorithm (VB-VBC) and approximate variational approximations (AVB-VBC). Black dotted line refers to zero. Results are averaged across 40 replicate datasets.}
    \label{fig:timing_enumeration}
\end{figure}

\textbf{Posterior Consistency.} To evaluate empirically whether we can expect posterior consistency of the full VB algorithms of the four illustrative models, we conduct an additional simulation exercise where settings up to $n = 200,000$ are considered. We collect root mean square errors (RMSEs) of all model parameters relative to their true values as well as the average estimated posterior variance. Fig. \ref{fig:consistency} and \ref{fig:consistency_log} in the supplementary material report these measures and show that clearly and across all settings, posterior variability as well as RMSEs quickly approach zero. This provides empirical evidence that posterior consistency holds for the four considered models.

\section{Applications to Real-World Data}
\label{sec:applications}

To illustrate the methodology and explore the performance of the algorithms in real-world settings, we consider two empirical exercises. First, we use VB and AVB algorithms to fit probit models to survey data with $n=41,587$ observations and $p=74$ covariates. The data is compiled based on harmonized Afrobarometer surveys\footnote{The Afrobarometer is a pan-African research network that measures public attitudes in Africa using face-to-face surveys, see \texttt{https://www.afrobarometer.org/}.} conducted in 35 African countries. The outcome of interest is a binary indicator measuring whether the surveyed individual has intentions to migrate internationally. Explanatory variables include various potential drivers of migration intentions, as well as country indicators, socio-economic and demographic characteristics of the survey respondent, among others. For detailed information on the dataset, refer to \citet{hoffmann2024interrelated} where this data is introduced. 

For the second illustration, we fit STAR models to count-valued data with $n = 8,742$ observations and $p = 43$ covariates. The counts represent register-based long-run origin-destination specific migration flows between 94 districts of Austria in the 10-year period 2014-2023. As explanatory variables for these flows, we consider classical migration determinants, including distance and contiguity between regions, as well as population counts, age, sex and educational composition of the populations, availability of childcare facilities, unemployment rates, and sectoral economic composition in both origin and destination regions, among others. All of these variables are measured for the baseline year 2013, are publicly available, and have been sourced from the online data explorer of the Austrian national statistical office (\textit{Statistik Austria}).

These applications were selected as they relate to typical problems in applied research and cover both count and binary outcomes. Importantly, they also illustrate the behavior of the algorithms in two rather extreme scenarios. For the Afrobarometer data, the combination of a moderately large number of covariates, an intricate correlation structure between the covariates and little likelihood information per observation due to the binary setting creates a scenario where we cannot assume that asymptotic results will apply, as the information available to estimate each parameter is relatively low. We thus expect the overall likelihood information to be relatively weak and AVB to be conservative relative to VB. For the migration count data, likelihood information is stronger, with a mean count of approximately 283 and only 2\% zero counts in the dataset. In this case, we expect essentially equivalent behavior between the VB and AVB algorithms, as the likelihood information in $p(y_i \mid z_i)$ is quite strong. 

We use the model space exploration algorithm described in Section~\ref{sec:algorithmic} to screen 250,000 models after an initial burn-in period of 50,000 models. For brevity, we do not describe the specific coefficient estimates and potential theoretical channels that may explain the estimated relationships, but focus on the variable selection and model averaging results instead. The time differences between the AVB and VB algorithms in these applications are substantial. For the Afrobarometer data, the full VB probit algorithm took more than two days to screen 300,000 models, while the AVB probit algorithm took two hours (a reduction of 96\% in runtime). For the migration data, the VB STAR algorithm took 6.5 hours versus 20 minutes for the AVB algorithm (a reduction of 95\%).

\begin{figure}[t!]
    \centering
    \begin{subfigure}[t]{0.8\linewidth}
        \centering
        \includegraphics[width=\linewidth]{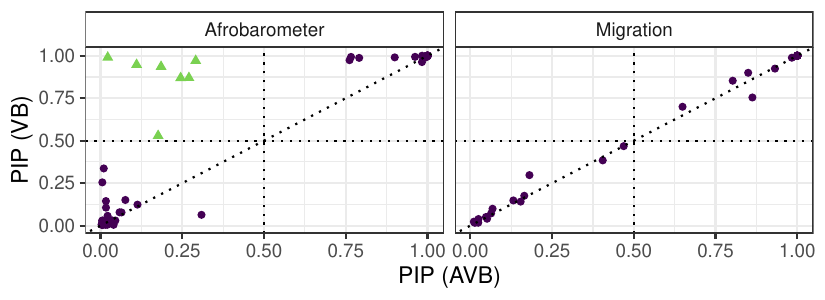}
        \caption{Posterior Inclusion Probabilities.}
    \end{subfigure}\\
    \begin{subfigure}[t]{0.8\linewidth}
        \centering
        \includegraphics[width=\linewidth]{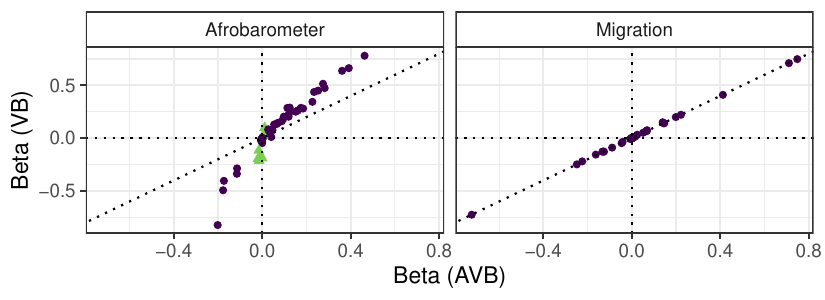}
        \caption{Posterior Means $\bm{\beta}$.}
    \end{subfigure}
    \caption{Panel (a) shows posterior inclusion probability estimates. Panel (b) shows model-averaged posterior mean estimates for the regression coefficients $\bm{\beta}$. Results are based on 250,000 sampled models. Green triangles indicate coefficients where the median probability models obtained from AVB and VB disagree.}
    \label{fig:app_res}
\end{figure}

Despite the significant time difference, the model averaging results are similar between both methods. Fig.~\ref{fig:app_res} compares posterior inclusion probabilities and model-averaged posterior mean estimates of $\bm{\beta}$ of the AVB and VB algorithms. For the Afrobarometer application, the highest scoring model in terms of VBC and the median probability model under AVB and VB agree on inclusion/exclusion of 67 variables, while VB includes 7 variables that are excluded under the AVB results. For the migration data, the highest scoring models agree on 41 inclusion/exclusion patterns, while VB includes 2 more variables that are excluded by AVB. The median probability models are equivalent under VB and AVB for the migration data. These results are also reflected in posterior model size estimates, which for the Afrobarometer data are 34.48 (SD: 1.53) for AVB and 40.77 (SD: 1.39) for VB, and for the migration data 24.00 (SD: 1.53) for AVB versus 24.23 (SD: 1.54) for VB. 

The few variables where disagreement between AVB and VB occurs are typically not those that have parameter estimates clearly close to zero or decidedly away from zero -- on those both methods tend to agree, see Panel (b) in Fig.~\ref{fig:app_res}. Panel (b) also clearly shows the shrinkage towards $\bm{\beta}=0$ built into the AVB algorithm by construction when the likelihood is rather weak (Afrobarometer data), while estimates agree almost perfectly in cases with strong likelihood (migration data). Disagreement typically happens for coefficients that are neither clearly predictive nor clearly non-predictive for the outcome. Here, the shrinkage of AVB towards zero is stronger than the likelihood information that pulls the coefficients away from zero. Importantly, variables that are clearly included by AVB are also clearly included by VB. Overall, these results are thus in line with theoretical expectations. In Sec.~\ref{app:additional_sim} in the supplementary material, we replicate this behavior empirically and provide further explanations in the context of an additional simulation experiment. Further heuristics on the relationship between AVB and VB are provided in supplementary Sec.~\ref{app:proof_asymptotic_shrinkage}. 

\section{Discussion \& Concluding Remarks}
\label{sec:conclusion}

This article introduces a novel framework for Bayesian variable selection and model averaging in latent variable regression models under large samples. We leverage mean-field variational posterior approximations and consider a new marginal likelihood approximation, showing that it is both computationally tractable and asymptotically consistent for a broad class of latent regression settings. We further propose an AVB algorithm, which fixes the latent outcomes at an initial null-model estimate. AVB is demonstrated to greatly reduce computation times across large model spaces, while approximately and locally preserving the asymptotic model selection properties of the full variational Bayes approach to first order. Extensive simulations and real-data applications confirmed that the proposed approximations deliver substantial gains in computational efficiency with negligible loss in accuracy for large sample model selection or averaging.

Several directions for future research emerge from this work. Our theoretical framework currently addresses asymptotic settings where the number of potential regressors $p$ remains fixed as sample size $n$ grows. Extending to high-dimensional scenarios where $p$ and $n$ diverge at comparable rates presents both theoretical and computational challenges, particularly regarding posterior contraction and model selection consistency. These extensions would require careful analysis of how our approximations behave when the model space grows with the sample size.

The assumptions we made also raise questions about guarantees for asymptotic parameter consistency in latent variable regression models more broadly. While our empirical results suggest parameter consistency in the model families we consider here, formal consistency proofs remain to be developed for the general class of latent variable regression models. Recent work provides promising directions - for instance, \citet{hall2011asymptotic} and \citet{batardiere2024evaluating} provide multiple relevant findings for PLN models that could potentially be extended to a broader class of models. Such theoretical developments would strengthen the theoretical foundation of several of our results.

From a methodological perspective, the finite-sample behavior of AVB when signals are weak suggests several refinements worth exploring. These include partial updating schemes for latent outcomes, data-driven diagnostics to assess approximation quality, and ex-post correction factors derived from null models. Our supplementary results provide starting points for developing such enhancements. Additionally, while we focused on unit information $g$-priors, extending to more flexible prior specifications through non-degenerate $g \sim p(g)$ distributions (\citealp{liang_etal_08}; \citealp{leysteel2012}) may improve practical model selection and averaging performance. Moreover, latent variable representations with multiple auxiliary variable layers - such as those discussed for logistic models in \citealp{zens2024ultimate} - are not directly covered by our analysis and would require further investigation.

Finally, the framework naturally extends to richer model classes including multivariate latent regression models like multinomial probit, models with latent thresholds such as ordered probit, and other latent variable formulations including factor and mixture models. Each extension would require adapting our theoretical results and algorithms to the specific structure of these models, but we believe the core principles we establish here provide an important  foundation for such developments.

\spacingset{1.6}

\bibliographystyle{agsm}
\bibliography{AVB}

\spacingset{1.9} 

\clearpage
\begin{center}
{\large\bf SUPPLEMENTARY MATERIAL}
\end{center}

\setcounter{figure}{0}
\setcounter{section}{0}
\setcounter{table}{0}
\setcounter{equation}{0}

\renewcommand\thesection{S\arabic{section}}
\renewcommand{\theHsection}{S\arabic{section}}
\renewcommand\theequation{S\arabic{equation}}
\renewcommand\thefigure{S\arabic{figure}}
\renewcommand\thetable{S\arabic{table}}

\section{Relation of VBC approximation and ELBO}
\label{app:elbo_vs_vbc}

The \textit{evidence lower bound} (ELBO) is a key quantity in variational inference. In our setting, it is given by 
\begin{equation}
\label{eq:elbo}
    \text{ELBO}= \mathbb{E}_{q(\bm{z})q(\bm{\theta})}\left[\log\frac{p(\bm{y}|\bm{z})\, p(\bm{z}|\bm{\theta})\, p(\bm{\theta})}{q(\bm{z})q(\bm{\theta})}\right],
\end{equation}
and recent evidence has shown that the ELBO can be a useful and consistent asymptotic model selection criterion, recovering the BIC asymptotically in the sample size $n$, but typically leading to closer marginal likelihood approximations (\citealp{zhang2024bayesian}). Both $\log p_{\text{VBC}}(\bm{y})$ and $\text{ELBO}$ can be written in terms of the true log marginal likelihood $p(\bm{y})$ and a `residual term':
\begin{equation}
\begin{split}
    \log p(\bm{y}) &= \log p_{\text{VBC}}(\bm{y}) + \log \frac{q(\hat{\bm{z}})q(\hat{\bm{\theta}})}{p(\hat{\bm{z}}, \hat{\bm{\theta}}|\bm{y})}\\
    \log p(\bm{y}) &= \text{ELBO} + \mathbb{E}_{q(\bm{z})q(\bm{\theta})}\left[\log \frac{q(\bm{z})q(\bm{\theta})}{p(\bm{z}, \bm{\theta}|\bm{y})}\right],
\end{split}
\end{equation}
from which can be seen that, while the $\text{ELBO}$ deviates from the true marginal likelihood due to the Kullback-Leibler divergence between the true posterior distribution and the variational approximation, the $p_{\text{VBC}}(\bm{y})$ approximation deviates from the true marginal likelihood due to discrepancies between the true posterior probability density function and its variational approximation at the specific point $(\hat{\bm{z}}, \hat{\bm{\theta}})$. In the general case, the ELBO will thus be different from $p_{\text{VBC}}(\bm{y})$. In our empirical exercises, we find that $p_{\text{VBC}}(\bm{y})$ typically leads to slightly improved model selection and model averaging performance compared to the $\text{ELBO}$.

Asymptotically, and under the conditions stated in Sec.~\ref{sec:general}, the VBC marginal likelihood approximation approaches 
\begin{equation}
\begin{split}
    \log p_{\text{VBC}}(\bm{y}) \xrightarrow[n\rightarrow\infty]{} &\log p(\bm{y}|\hat{\bm{z}}) + \log p(\hat{\bm{z}}|\hat{\bm{\theta}}) + \log p(\hat{\bm{\theta}}) - \log q(\hat{\bm{z}})
 +\\&\frac{d}{2}\log(2\pi) - \frac{d}{2}\log(n)+\frac{1}{2}\log|\bm{\Omega}|,
\end{split}
\end{equation}
see (\ref{eq:asymptotic_vbc}). For comparison, the ELBO behaves asymptotically as
\begin{equation}
\begin{split}
    \text{ELBO} \xrightarrow[n\rightarrow\infty]{} &\mathbb{E}_{q(\bm{z})}[\log p(\bm{y}|\bm{z})] + \mathbb{E}_{q(\bm{z})}[\log p(\bm{z}|\hat{\bm{\theta}})] + \log p(\hat{\bm{\theta}}) - \mathbb{E}_{q(\bm{z})}[\log q(\bm{z})] + \\&\frac{d}{2} \log(2\pi e) - \frac{d}{2}\log(n) + \frac{1}{2} \log | \bm{\Omega}|.
\end{split}    
\end{equation}

\section{CAVI Updates}
\label{sec:cavi_updates}

In this section, we take $q(z_i)$ as given and consider updates for any given model with $p$ covariates. To determine the variational parameters \( \mu_{\alpha} \), \( \omega_{\alpha} \), \( \bm{\mu}_{\beta} \), \( \bm{\Omega}_{\beta} \), \( a \), and \( b \), we define \( \mathbb{E}_{q(z_i)}\left[z_i\right] = m_i \) and \( \mathbb{V}_{q(z_i)}\left[z_i\right] = s_i \), representing the expected value and variance of $z_i$ under variational density \( q(z_i) \), respectively.

\paragraph{Updating $q(\alpha)$}

To update the variational density of the intercept term $\alpha$, we consider all terms involving $\alpha$ in the joint log density $\log p(\bm{y}, \bm{z}, \alpha, \bm{\beta}, \sigma^2)$:
\[
\log q(\alpha) \propto \mathbb{E}_{-q(\alpha)} \left[ \log p(\bm{z} | \alpha, \bm{\beta}, \sigma^2) \right] \propto \mathbb{E}_{-q(\alpha)} \left[ -\frac{1}{2\sigma^2} \sum_{i=1}^n (z_i - \alpha - \bm{x}_i^\top \bm{\beta})^2 \right].
\]

After expanding, taking expectations, and rearranging, we obtain
\[
q(\alpha) \propto \mathcal{N}\left( \bar{m}, \frac{b}{n a} \right),
\]
where $\bar{m}=\frac{1}{n}\sum_i m_i$ and we have utilized the facts that $\sum_i \bm{x}_i=\bm{0}$ due to the centering of covariates and that $\mathbb{E}\left[ \frac{1}{\sigma^2} \right]=\frac{a}{b}$ corresponds to the expectation of a Gamma-distributed variable.

\paragraph{Updating $q(\bm{\beta})$}

Similarly, for the coefficients $\bm{\beta}$, we have
\[
\begin{aligned}
\log q(\bm{\beta}) &\propto \mathbb{E}_{-q(\bm{\beta})} \left[ \log p(\bm{z} | \alpha, \bm{\beta}, \sigma^2) + \log p(\bm{\beta} | \sigma^2) \right]\\ &\propto \mathbb{E}_{-q(\bm{\beta})} \left[ -\frac{1}{2\sigma^2} \sum_{i=1}^n (z_i - \alpha - \bm{x}_i^\top \bm{\beta})^2 - \frac{1}{2} \bm{\beta}^\top \left( \frac{1}{g \sigma^2} \bm{X}^\top \bm{X} \right) \bm{\beta} \right].
\end{aligned}
\]

After expanding, taking expectations, and rearranging, we obtain
\[
\begin{aligned}
q(\bm{\beta}) &\propto \mathcal{N}\left( \delta (\bm{X}^\top \bm{X})^{-1} \bm{X}^\top \bm{m}, \delta \frac{b}{a} (\bm{X}^\top \bm{X})^{-1} \right),
\end{aligned}
\]
where $\delta = \frac{g}{1 + g}$ and $\bm{m} = (m_1, \dots, m_n)^\top$.

\paragraph{Updating $q(\sigma^2)$}

Finally, to derive the variational update for $q(\sigma^2)$, we consider
\[
\begin{aligned}
    \log q(\sigma^2) &\propto \mathbb{E}_{-q(\sigma^2)} \left[ \log p(\bm{z} | \alpha, \bm{\beta}, \sigma^2) + \log p(\bm{\beta} | \sigma^2) + \log p(\sigma^2) \right]\\
    &\propto \mathbb{E}_{-q(\sigma^2)} \Bigg[-\frac{n}{2}\log \sigma^2 -\frac{1}{2\sigma^2} \sum_{i=1}^n (z_i - \alpha - \bm{x}_i^\top \bm{\beta})^2\\
    &\qquad\qquad\quad - \frac{p}{2}\log \sigma^2 - \frac{1}{2} \bm{\beta}^\top \left( \frac{1}{g \sigma^2} \bm{X}^\top \bm{X} \right) \bm{\beta} - \log \sigma^2 \Bigg].
\end{aligned}
\]

After taking expectations, dropping constant terms in $\sigma^2$ and rearranging, the variational density for $\sigma^2$ can be shown to be
\[
\begin{aligned}
    q(\sigma^2) &\propto \mathcal{IG} \left(a, b \right)\\
    a &= \frac{n + p}{2}\\
    b &= \frac{1}{2}\left[\sum_{i=1}^n \Bigl(m_i - \mu_\alpha - \bm{x}_i^\top\bm{\mu}_\beta\Bigr)^2 + \sum_{i=1}^n s_i + n\omega_\alpha +\right. \\[1ex]
    &\left.\quad + \frac{1}{g}\Bigl(\bm{\mu}_\beta^\top \bm{X}^\top \bm{X}\bm{\mu}_\beta\Bigr) + (1+\frac{1}{g})\operatorname{tr}\Bigl(\bm{X}^\top \bm{X}\bm{\Omega}_\beta\Bigr)\right],
\end{aligned}
\]
where the first term in $b$ is the classical sum of squared residuals and the remaining quantities account for uncertainty in $z_i$, $\alpha$ and $\bm{\beta}$. Some of these quantities collapse to simple terms that can be very efficiently computed under the conjugate $g$-prior setup.

\paragraph{Summary.}

We summarize the variational parameter updates as follows, suppressing model indices and defining $\bm{m}=(m_1,\dots,m_n)^\top$:
\begin{equation}
\begin{aligned}
\mu_\alpha &= \frac{1}{n}\sum_{i=1}^n m_i, \quad 
\omega_\alpha = \frac{b}{na}, \\[1ex]
\bm{\mu}_\beta &= \frac{g}{1+g}\,(\bm{X}^\top \bm{X})^{-1}\bm{X}^\top \bm{m}, \quad 
\bm{\Omega}_\beta = \frac{g}{1+g}\,\frac{b}{a}\,(\bm{X}^\top \bm{X})^{-1}, \\[1ex]
a &= \frac{n+p}{2}, \\[1ex]
b &= \frac{1}{2}\left[\sum_{i=1}^n \Bigl(m_i - \mu_\alpha - \bm{x}_i^\top\bm{\mu}_\beta\Bigr)^2 + \sum_{i=1}^n s_i + n\omega_\alpha +\right. \\[1ex]
    &\left.\quad + \frac{1}{g}\Bigl(\bm{\mu}_\beta^\top \bm{X}^\top \bm{X}\bm{\mu}_\beta\Bigr) + (1+\frac{1}{g})\operatorname{tr}\Bigl(\bm{X}^\top \bm{X}\bm{\Omega}_\beta\Bigr)\right].
\end{aligned}
\end{equation}


\newpage
\section{Proofs}
\label{app:proofs}

We start by stating two results that we will repeatedly refer to:

\begin{lemma}[Score--variance identity]\label{lem:score-var-identity}
Fix $i$ and let $Z_i\sim\mathcal N(\eta_i,\sigma^2)$ with likelihood contribution $p(y_i\mid z_i)$. Define
\[
m_i(\eta_i):=\E[Z_i\mid Y_i=y_i,\eta_i].
\]
Then, for all $\eta_i$,
\[
\frac{\partial m_i}{\partial\eta_i}(\eta_i)=\frac{\mathbb{V}(Z_i\mid Y_i=y_i;\eta_i)}{\sigma^2}.
\]
\end{lemma}

\begin{proof}
Write the unnormalized posterior $\tilde\pi(z_i):=\phi_\sigma(z_i-\eta_i)\,L(z_i)$ with $L(z_i):=p(y_i\mid z_i)$ and normalizing constant  $C(\eta_i):=\int \tilde\pi(u_i)\,du_i$, where $\phi_\sigma(\cdot)$ is the pdf of a Gaussian with mean zero and variance $\sigma^2$. Then, by the definition of the expected value,
\[
m_i(\eta_i)=\frac{\int z_i\,\tilde\pi(z_i)\,dz_i}{C(\eta_i)}.
\]
By the properties of Gaussian densities $\partial_{\eta_i}\log\phi_\sigma(z_i-\eta_i)=(z_i-\eta_i)/\sigma^2$. Due to the chain rule, we have $\partial_\theta f(\theta)=f(\theta)\,\partial_\theta \log f(\theta)$ for any positive $f(\theta)$. Hence, it follows that
$\partial_{\eta_i} \tilde\pi(z_i)=\tilde\pi(z_i)\,(z_i-\eta_i)/\sigma^2$. Differentiating $m_i(\eta_i)$ via the quotient rule, replacing $\tilde\pi(z_i)\,(z_i-\eta_i)/\sigma^2$ with the partial derivative $\partial_{\eta_i} \tilde\pi(z_i)$ and writing expectations under the normalized posterior $\pi(z_i):=\tilde\pi(z_i)/C(\eta_i)$ gives
\begin{align*}
\frac{\partial m_i}{\partial\eta_i}
&=\frac{1}{\sigma^2}\big(\E_\pi[Z_i^2]-\eta_i\,\E_\pi[Z_i]\big)
-\E_\pi[Z_i]\cdot \frac{1}{\sigma^2}\big(\E_\pi[Z_i]-\eta_i\big)\\[3pt]
&=\frac{1}{\sigma^2}\Big(\E_\pi[Z_i^2]-\eta_i\,\E_\pi[Z_i]-\E_\pi[Z_i]^2+\eta_i\,\E_\pi[Z_i]\Big)\\[3pt]
&=\frac{1}{\sigma^2}\Big(\E_\pi[Z_i^2]-\E_\pi[Z_i]^2\Big)
=\frac{\mathbb{V}_\pi(Z_i)}{\sigma^2}.
\end{align*}
\end{proof}

\begin{lemma}[Weight bounds under log-concave likelihoods]\label{lem:weight-bounds}
Under the setup of Lemma~\ref{lem:score-var-identity}, assume that $p(y_i\mid z_i)$ is univariate and either (i) an indicator function with respect to a convex interval (as in probit, tobit, STAR models) or (ii) a $C^2$ log-concave pmf or pdf (as in many common settings, including for example Poisson likelihoods with mean $\lambda=e^{z_i}$). Let
\[
w_{0i}:=\left.\frac{\partial m_i}{\partial\eta_i}\right.
\]
Then $0<w_{0i}\leq1$.
\end{lemma}

\begin{proof}
For interval functions on a convex set, the result immediately follows from the definition of $w_{0i}$ derived in Lemma~\ref{lem:score-var-identity}, as the posterior of $z_i$ is a truncated version of its Gaussian prior, with posterior variance necessarily smaller than prior variance $\sigma^2$ due to truncation. For log-concave densities, write the posterior density as $\pi(z_i)\propto e^{-R(z_i)}$ with $R(z_i)=\frac{(z_i-\eta_i)^2}{2\sigma^2}-\log L(z_i)$. Note that $L$ being log-concave implies $\log L''\le 0$, and $R''(z_i)\ge 1/\sigma^2$ follows. Now for a random variable $Z$, if $Z\sim \pi(Z)\propto e^{-R(Z)}$ with smooth $R$, the one-dimensional Brascamp-Lieb/Poincaré inequality gives, for smooth $f$, $\mathbb{V}(f(Z))\le \mathbb{E}\!\left[\frac{(f'(Z))^2}{R''(Z)}\right]$. Applying this with $f(Z)=Z$ yields $\mathbb{V}(Z)\le \mathbb{E}[1/R''(Z)]\le \sigma^2$. By Lemma~\ref{lem:score-var-identity}, and using $\mathbb{V}(z_i)$ as shorthand for the posterior variance of $z_i$, we thus have $w_{0i}=\mathbb{V}(z_i)/\sigma^2\le 1$. With $\sigma^2>0$ and any nondegenerate log-concave likelihood contribution $p(y_i|z_i)$, we have $\mathbb{V}(z_i)>0$ and thus $w_{0i}>0$, proving the result.
\end{proof}

\subsection{Proposition 1: VBC selects the true model asymptotically}
\label{app:proof_vbc_consistency}
Given condition (v) is fulfilled, it is easy to show from (\ref{eq:approx2}) that $\log p_{\text{VBC}}(\bm{y}\mid\mathcal{M}_k)$ asymptotically collapses to the BIC form  
\begin{equation}
\label{eq:bic_form}
    \log p_{\text{VBC}}(\bm{y}\mid\mathcal{M}_k)
=\log p(\bm{y}\mid\hat{\bm\theta}_{k})
-\frac{d_k}{2}\log n+\mathcal{O}(1)
\end{equation}
at the evaluation point $(\hat{\bm{z}}_k,\hat{\bm{\theta}}_k)$. The rest of the proof then follows along the lines of showing asymptotic model selection consistency for the BIC given consistent parameter estimates. From (\ref{eq:bic_form}), for any two models $k$ and $j$, the difference in their VBC values is asymptotically
\[
\mathrm{VBC}_k - \mathrm{VBC}_j = -2\log\frac{p(\bm{y}|\hat{\bm{\theta}}_k)}{p(\bm{y}|\hat{\bm{\theta}}_j)} + (d_k - d_j)\log n + \mathcal{O}(1).
\]

For the true model $k^*$ versus any misspecified model $k \neq k^*$, we have three cases: If model $k$ is nested within $k^*$ (underspecified), the likelihood ratio term dominates with $-2\log\frac{p(\bm{y}|\hat{\bm{\theta}}_k)}{p(\bm{y}|\hat{\bm{\theta}}_{k^*})} = \mathcal{O}(n)$, overwhelming the penalty difference. If model $k$ contains $k^*$ (overspecified), then $d_k > d_{k^*}$ and the penalty term $(d_k - d_{k^*})\log n$ dominates asymptotically. For non-nested models, the true model achieves the highest likelihood asymptotically by assumption. Thus, $\Pr\{\mathrm{VBC}_{k^*} < \mathrm{VBC}_k \text{ for all } k \neq k^*\} \to 1$ as $n \to \infty$.

\subsection{Proposition 2: Asymptotic local shrinkage relationship of AVB and VB}
\label{app:proof_asymptotic_shrinkage}

We will prove the result in two steps. First, we derive a local linear map between AVB and VB parameter estimates at $\bm{\beta}=0$. Second, we show that this map asymptotically collapses to a scalar shrinkage relationship between AVB and VB.

\begin{lemma}[Local linear relationship]\label{lem:local-linear-danskin}
Assume for each $i$, the posterior mean $m_i(\eta):=\E[z_i\mid y_i,\eta]$ is $C^1$ at the null predictor $\eta=\alpha_0$, and define
\[
w_{0i}\ :=\ \left.\frac{\partial m_i}{\partial\eta}\right|_{\eta=\alpha_0}
=\frac{\mathbb{V}(z_i\mid y_i;\,\eta=\alpha_0)}{\sigma^2}\in(0,1],
\qquad
\bm{W}_0:=\diag(w_{01},\ldots,w_{0n}),
\]
using Lemma \ref{lem:score-var-identity}. Let $\bm{X}\in\R^{n\times p}$ have centered columns ($\bm{X}^\top\bm{\iota}=0$) and full column rank. Let $\kappa:=1+\tfrac1g>1$ and let $\hat{\bm\beta}_{\mathrm{VB}}$ and $\hat{\bm\beta}_{\mathrm{AVB}}$ be the maximizers of the profiled ELBOs in $\bm\beta$ for VB and AVB, with VB intercept $\alpha_{\mathrm{VB}}$. Then the VB/AVB first-order conditions admit the local linearization
\[
\hat{\bm\beta}_{\mathrm{AVB}}
=\Big[\bm{I}-\tfrac1\kappa\,(\bm{X}^\top \bm{X})^{-1}\bm{X}^\top \bm{W}_0 \bm{X}\Big]\hat{\bm\beta}_{\mathrm{VB}}
-\frac{\alpha_{\mathrm{VB}}-\alpha_0}{\kappa}\,(\bm{X}^\top \bm{X})^{-1}\bm{X}^\top \bm{W}_0\bm{\iota}
+\bm r_n,
\]
where $\bm r_n=o_p\!\big(\|\hat{\bm\beta}_{\mathrm{VB}}\|+|\alpha_{\mathrm{VB}}-\alpha_0|\big)$. Moreover,
\[
\bm{M}\ :=\ \bm{I}-\tfrac1\kappa\,(\bm{X}^\top \bm{X})^{-1}\bm{X}^\top \bm{W}_0 \bm{X}
\]
is a contraction in the $\bm{X}^\top \bm{X}$-geometry (its eigenvalues lie in $(0,1]$).
\end{lemma}

\begin{proof}
Consider a mean-field variational approximation with Gaussian factor
\(
q_\beta(\bm{\beta})=\mathcal{N}(\bm{\mu}_\beta,\bm{\Omega}_\beta)
\)
and remaining factors
\(
(q_\alpha,\,q_{\sigma^2},\,q_{\bm z})
\).
Write \(\lambda:=(q_\alpha,q_{\sigma^2})\).
Let the ELBO be
\[
\mathcal{L}(\bm{\mu}_\beta,\bm{\Omega}_\beta,\lambda,q_{\bm z}).
\]
We profile this ELBO in \(\bm{\mu}\equiv\bm{\mu}_\beta\):
\begin{align}
\text{VB (full profiling):}\quad
\mathcal{L}^*(\bm{\mu})
&:= \sup_{\bm{\Omega}_\beta,\ \lambda,\ q_{\bm z}}\ \mathcal{L}(\bm{\mu},\bm{\Omega}_\beta,\lambda,q_{\bm z}),\\[2mm]
\text{AVB (fix $q_{\bm z}$ at null):}\quad
\bar{\mathcal{L}}(\bm{\mu})
&:= \sup_{\bm{\Omega}_\beta,\ \lambda}\ \mathcal{L}(\bm{\mu},\bm{\Omega}_\beta,\lambda,\bar q_{\bm z}),
\end{align}
where \(\bar q_{\bm z}\) is the optimizer of \(q_{\bm z}\) in the \emph{full} VB problem at \(\bm{\mu}=\bm{0}\) (the null fit). A direct differentiation using Danskin's theorem yields
\begin{equation}
\label{eq:grad-VB}
\nabla_{\bm{\mu}} \mathcal{L}^*(\bm{\mu})
= \bm{X}^\top \bm{m}^*(\bm{\mu}) \;-\; \kappa\,\bm{X}^\top \bm{X}\,\bm{\mu},
\qquad \kappa:=1+\frac{1}{g}>1.
\end{equation}

in the VB case and

\begin{equation}
\label{eq:grad-AVB}
\nabla_{\bm{\mu}} \bar{\mathcal{L}}(\bm{\mu})
= \bm{X}^\top \bar{\bm m} \;-\; \kappa\,\bm{X}^\top \bm{X}\,\bm{\mu}.
\end{equation}

in the AVB case. By construction, \(\hat q_{\bm z}(\bm{0})=\bar q_{\bm z}\), hence \(\bm{m}^*(\bm{0})=\bar{\bm m}\) and
\(
\nabla_{\bm{\mu}} \mathcal{L}^*(\bm{0})=\nabla_{\bm{\mu}} \bar{\mathcal{L}}(\bm{0})=X^\top \bar{\bm m}.
\)
Hence, the profiled gradients satisfy
\(
\nabla_{\bm\mu}\mathcal L^*(\bm\mu)=\bm{X}^\top \bm m^*(\bm\mu)-\kappa \bm{X}^\top \bm{X}\,\bm\mu
\)
and
\(
\nabla_{\bm\mu}\bar{\mathcal L}(\bm\mu)=\bm{X}^\top \bar{\bm m}-\kappa \bm{X}^\top \bm{X}\,\bm\mu,
\)
with components $m_i(\eta_i)$, $\eta_i=\alpha+\bm x_i^\top\bm\mu$, and $\bar{\bm m}$ from the null ($\eta_i=\alpha=\alpha_0$). A first-order expansion at $\eta_i=\alpha_0$ gives
\(
m_i(\eta_i)=m_i(\alpha_0)+w_{0i}\{(\alpha-\alpha_0)+\bm x_i^\top\bm\mu\}+o(|\alpha-\alpha_0|+\|\bm\mu\|).
\)
Stacking and premultiplying by $\bm{X}^\top$,
\[
\bm{X}^\top \bm m^*(\bm\mu)
= \bm{X}^\top \bar{\bm m}+\bm{X}^\top \bm{W}_0\{(\alpha-\alpha_0)\bm{\iota}+\bm{X}\bm\mu\}
+o(|\alpha-\alpha_0|+\|\bm\mu\|).
\]
Set the VB first order condition (FOC) $0=\bm{X}^\top \bm m^*(\hat{\bm\mu}_{\mathrm{VB}})-\kappa \bm{X}^\top \bm{X}\hat{\bm\mu}_{\mathrm{VB}}$ and the AVB FOC $0=\bm{X}^\top \bar{\bm m}-\kappa \bm{X}^\top \bm{X}\hat{\bm\mu}_{\mathrm{AVB}}$, subtract, and solve for $\hat{\bm\mu}_{\mathrm{AVB}}$ to obtain the desired result. Since $0\preceq \bm{W}_0\preceq \bm{I}$, we have $0\preceq \bm{X}^\top \bm{W}_0 \bm{X}\preceq \bm{X}^\top \bm{X}$, hence the spectrum claim.
\end{proof}

This relationship reveals that, in a local first order approximation, AVB estimates are essentially a generalized shrinkage version of VB estimates. The map between the two implies `leakage' via the off-diagonals of $M$ that potentially lead to rotations away from a scalar shrinkage relationship, as well as a drift term due to the mismatch between the intercept estimates. We find that, in practice and in large samples, this drift and rotation are weak and AVB is approximately a scalar shrinkage version of VB. More formally, we can show that, asymptotically in $n$, this result extends to the first-order approximate scalar shrinkage relationship of AVB and VB formulated in Proposition~\ref{prop:local-shrink-simple}. The following Lemma directly yields the proposition: 

\begin{lemma}[Local linear scalar shrinkage relationship asymptotically] Let the setup of Lemma~\ref{lem:local-linear-danskin} hold. Suppose, at the null ($\bm\beta=\mathbf 0$), $\bm{X}$ is independent of the innovation $\varepsilon$ in $\bm{z}=\alpha+\bm{X}\bm\beta+\varepsilon$ so that $\bm{X}$ is locally independent of $\bm{y}$ and of $\bm{W}_0$; assume $\max_i\|\bm x_i\|^2/n\to 0$, $\|\hat{\bm\beta}_{\mathrm{VB}}\|=o_p(1)$, and $|\alpha_{\mathrm{VB}}-\alpha_0|=o_p(1)$. Let $\E[\bm{W}_0]$ denote the expected value of the diagonal elements of $\bm{W}_0$. Then
\[
(\bm{X}^\top \bm{X})^{-1}\bm{X}^\top \bm{W}_0 \bm{X} \ \xrightarrow{p}\ \E[\bm{W}_0]\;\bm{I}_p,
\qquad
(\bm{X}^\top \bm{X})^{-1}X^\top \bm{W}_0\bm{\iota} \ \xrightarrow{p}\ 0,
\]
with $\E[\bm{W}_0]\in(0,1]$. Consequently,
\[
\;
\hat{\bm\beta}_{\mathrm{AVB}}
=\Big(1-\frac{\E[\bm{W}_0]}{\kappa}\Big)\,\hat{\bm\beta}_{\mathrm{VB}}
\;+\;o_p\!\bigl(\|\hat{\bm\beta}_{\mathrm{VB}}\|\bigr),
\;
\]
i.e., AVB is a first-order \emph{scalar shrinker} of VB with common (model-independent) factor $c:=1-\E[\bm{W}_0]/\kappa\in(0,1)$, and the intercept-induced drift vanishes to first order.
\end{lemma}

\begin{proof}
With local independence at the null, by LLN and centered columns,
\[
\frac{1}{n}\bm{X}^\top \bm{W}_0 \bm{X} \to_p \E[\bm{W}_0]\cdot \frac{1}{n}\bm{X}^\top \bm{X},
\qquad
\frac{1}{n}\bm{X}^\top \bm{W}_0\bm{\iota} \to_p \E[\bm{W}_0]\cdot \frac{1}{n}\bm{X}^\top \bm{\iota}=0.
\]
Premultiplying by $(\bm{X}^\top \bm{X})^{-1}$ yields the stated limits. Insert these into Lemma~\ref{lem:local-linear-danskin}; the drift term is $(\alpha_{\mathrm{VB}}-\alpha_0)\,(\bm{X}^\top \bm{X})^{-1}\bm{X}^\top \bm{W}_0\bm{\iota}$, which converges in probability to $0$ because $\alpha_{\mathrm{VB}}-\alpha_0=o_p(1)$ and $(\bm{X}^\top \bm{X})^{-1}\bm{X}^\top \bm{W}_0\bm{\iota}\to_p 0$ due to independence of $\bm{X}$ and $\bm{W}_0$ at the null and centered columns of $\bm{X}$ yielding $\bm{X}'\bm{\iota}=0$. The scalar limit for the slope map follows, with $0<\E[\bm{W}_0]\leq1$ following from Lemma~\ref{lem:weight-bounds}.
\end{proof}

\subsection{Proposition 3: VBC inherits VB model selection consistency up to first order}
\label{app:proof_avb_inherits_consistency}

We will first show that, asymptotically and approximately in a neighborhood of $\bm{\beta}=0$, AVB and VB lead to the same likelihood ordering across models $\mathcal{M}_j$. The stated result then follows from the fact that asymptotically, the only differences between model assessments are based on the likelihood ordering. The following lemma is stated for the observed-data log likelihood and is used below for deterministic-link models, where the leading VBC term coincides with the observed-data likelihood via condition (v), which is the only case where we have results for model selection consistency of VBC.

Let $\ell_k(\bm\theta_k):=\log p(\bm y\mid \bm\theta_k,\mathcal M_k)$, with
$\bm\theta_k=(\alpha,\bm\beta_k,\sigma^2)$, and let $(\alpha_0,\sigma_0^2)$ denote the null values.
Condition~(v) and (\ref{eq:asymptotic_vbc}) imply the VBC \emph{BIC-form} for both VB and AVB estimators:
\begin{equation}
\label{eq:S-VBIC}
-2\log p_{\mathrm{VBC},k}^{(\cdot)}(\bm y)
=
-2\,\ell_k\!\big(\bm\theta_k^{(\cdot)}\big) + d_k\log n + O_p(1),
\qquad (\cdot)\in\{\mathrm{VB},\mathrm{AVB}\},
\end{equation}
where $\bm\theta_k^{\mathrm{VB}}=(\hat\alpha_k,\hat{\bm\beta}_k,\hat\sigma_k^2)$ and
$\bm\theta_k^{\mathrm{AVB}}=(\tilde\alpha_k,\tilde{\bm\beta}_k,\tilde\sigma_k^2)$ are the respective plug-ins. We use the local, model-independent scalar-shrinkage relation (Proposition~\ref{prop:local-shrink-simple}):
\begin{equation}
\label{eq:S-shrink}
\tilde{\bm\beta}_k
= c\,\hat{\bm\beta}_k + o_p\!\big(\|\hat{\bm\beta}_k\|\big),
\qquad c\in(0,1),
\end{equation}
uniformly over $k$ in a neighborhood of $\bm\beta=\bm 0$. We also use the standard second-order expansion around the null:
\begin{equation}
\label{eq:S-taylor}
\ell_k(\alpha_0,\bm\beta_k,\sigma_0^2)-\ell_k(\alpha_0,\bm 0,\sigma_0^2)
= s_k^\top \bm\beta_k-\tfrac12\,\bm\beta_k^\top H_k\,\bm\beta_k + r_k(\bm\beta_k),
\qquad r_k(\bm\beta_k)=o(\|\bm\beta_k\|^2),
\end{equation}
with $H_k$ positive definite.

\begin{lemma}[Equivalent likelihood orderings under scalar shrinkage]
\label{lem:S-like-order}
Under \eqref{eq:S-shrink}--\eqref{eq:S-taylor}, for any models $k,\ell$,
\[
\ell_k(\alpha_0,\tilde{\bm\beta}_k,\sigma_0^2)-\ell_\ell(\alpha_0,\tilde{\bm\beta}_\ell,\sigma_0^2)
=
(2c-c^2)\,
\Big(\ell_k(\alpha_0,\hat{\bm\beta}_k,\sigma_0^2)-\ell_\ell(\alpha_0,\hat{\bm\beta}_\ell,\sigma_0^2)\Big)
+ o_p(1),
\]
with $(2c-c^2)\in(0,1)$ independent of $k,\ell$. In particular, pairwise likelihood orderings at $(\alpha_0,\sigma_0^2)$ are preserved with probability $\to1$ as $n\to\infty$.
\end{lemma}

\begin{proof}[Proof of Lemma~\ref{lem:S-like-order}]
Let $f_k(\bm\beta):=s_k^\top \bm\beta-\tfrac12\,\bm\beta^\top H_k\,\bm\beta$ be the quadratic part in \eqref{eq:S-taylor}. Its unique maximizer is $\bm\beta_k^\star=H_k^{-1}s_k$, and for any $c\in\mathbb R$,
\begin{equation}
\label{eq:S-parabola}
f_k(c\,\bm\beta_k^\star)
= c\,s_k^\top\bm\beta_k^\star - \tfrac12 c^2\,{\bm\beta_k^\star}^\top H_k \bm\beta_k^\star
= (2c-c^2)\,f_k(\bm\beta_k^\star).
\end{equation}
Assume $\hat\beta_k = H_k^{-1}s_k + o_p(n^{-1/2})$ uniformly in $k$.
Using \eqref{eq:S-shrink}, $\tilde{\bm\beta}_k = c\,\bm\beta_k^\star + o_p(\|\bm\beta_k^\star\|)$.
Plugging into \eqref{eq:S-taylor} at $(\alpha_0,\sigma_0^2)$ gives
\begin{align*}
&\ell_k(\alpha_0,\hat{\bm\beta}_k,\sigma_0^2)-\ell_k(\alpha_0,\bm 0,\sigma_0^2)
= f_k(\bm\beta_k^\star)+o_p(1),\\
&\ell_k(\alpha_0,\tilde{\bm\beta}_k,\sigma_0^2)-\ell_k(\alpha_0,\bm 0,\sigma_0^2)
= f_k(c\,\bm\beta_k^\star)+o_p(1).
\end{align*}
Apply \eqref{eq:S-parabola} and subtract the two models’ displays (the common baseline
$\ell(\alpha_0,\bm 0,\sigma_0^2)$ cancels) to obtain the claim, uniformly in $k,\ell$.
\end{proof}

\textbf{Proof of Proposition~\ref{prop:avb-inherits}.}
By Lemma~\ref{lem:S-like-order}, pairwise \emph{likelihood} orderings at $(\alpha_0,\sigma_0^2)$ are preserved (with probability $\to1$) between VB and AVB.
From \eqref{eq:S-VBIC}, for any $k,\ell$ within a fixed method $(\cdot)$,
\[
-2\log p_{\mathrm{VBC},k}^{(\cdot)}(\bm y) + 2\log p_{\mathrm{VBC},\ell}^{(\cdot)}(\bm y)
=
-2\Big(\ell_k(\bm\theta_k^{(\cdot)})-\ell_\ell(\bm\theta_\ell^{(\cdot)})\Big)
+ (d_k-d_\ell)\log n + O_p(1).
\]
Thus the \emph{VBC} pairwise ordering differs from the likelihood ordering only by the deterministic penalty $(d_k-d_\ell)\log n$ (identical for VB and AVB) plus an $O_p(1)$ remainder. Since likelihood orderings are preserved by Lemma~\ref{lem:S-like-order}, VBC orderings are preserved as well with probability $\to1$ as $n\to\infty$.
Because the VB-based VBC selector is consistent by Proposition~\ref{lem:vbc_consistency}, the AVB-based VBC selector is also locally and first-order model-selection consistent:
\[
\Pr\!\left\{\arg\min_{k}\mathrm{VBC}_k^{\mathrm{AVB}}=k^*\right\}\to1.
\]
\qedhere

\subsection{Proposition 4: ELBO and VBC equivalence in probit case}\label{app:lemma4}

For the probit model there is an exact equivalence between the ELBO and the VBC at $(\bm{\hat\theta}, \bm{\hat z})$. Let $\bm{\theta} = (\alpha, \bm{\beta})$. Note that $p(y_i\mid z_i)=1$ for every $z_i$ in the support of $q(\bm{z})$. Therefore, the ELBO can be written as
\[
\mathcal{L} = \mathbb{E}_{q(\bm{\theta})q(\bm{z})}\left[\log\frac{p(\bm{z}\mid\bm{\theta})p(\bm{\theta})}{q(\bm{\theta})q(\bm{z})}\right],
\]
which is equal to
\[
\mathcal{L} = \mathbb{E}_{q(\bm{\theta})q(\bm{z})}\left[\log p(\bm{z}\mid\bm{\theta})\right] + \mathbb{E}_{q(\bm{\theta})}\left[\log p(\bm{\theta})\right] - \mathbb{E}_{q(\bm{\theta})}\left[\log q(\bm{\theta})\right] - \mathbb{E}_{q(\bm{z})}\left[\log q(\bm{z})\right].
\]

All the densities involved are either Gaussian or truncated Gaussian. For Gaussians, the log probability density functions are at most quadratic in the parameters of interest, which permits an exact second-order Taylor expansion. This property also applies to truncated Gaussians provided that the expansion point is chosen where the log pdf is twice continuously differentiable (thus excluding evaluations at the mode if it coincides with a truncation boundary). For the VBC, we evaluate at the expected value of the variational density. For a truncated Gaussian, this expectation necessarily lies in the region where the log pdf is twice differentiable and quadratic as long as we have positive variance, ensuring that the second-order Taylor expansion is exact at this point. 

Let $\hat{\bm{z}} = \mathbb{E}_{q(\bm{z})}[\bm{z}]$ and $\hat{\bm{\theta}} = \mathbb{E}_{q(\bm{\theta})}[\bm{\theta}]$. Expanding each term separately around $(\hat{\bm{z}}, \hat{\bm{\theta}})$ yields
\[
\begin{aligned}
\log p(\bm{z}\mid\bm{\theta}) &= \log p(\hat{\bm{z}}\mid\hat{\bm{\theta}}) + \nabla_{\bm{z},\bm{\theta}}\log p(\hat{\bm{z}}\mid\hat{\bm{\theta}})^\top \begin{pmatrix} \bm{z} - \hat{\bm{z}} \\ \bm{\theta} - \hat{\bm{\theta}} \end{pmatrix} \\
&\quad + \frac{1}{2}\begin{pmatrix} \bm{z} - \hat{\bm{z}} \\ \bm{\theta} - \hat{\bm{\theta}} \end{pmatrix}^\top \nabla^2_{\bm{z},\bm{\theta}}\log p(\hat{\bm{z}}\mid\hat{\bm{\theta}}) \begin{pmatrix} \bm{z} - \hat{\bm{z}} \\ \bm{\theta} - \hat{\bm{\theta}} \end{pmatrix},
\end{aligned}
\]
\[
\begin{aligned}
\log p(\bm{\theta}) &= \log p(\hat{\bm{\theta}}) + \nabla_{\bm{\theta}}\log p(\hat{\bm{\theta}})^\top (\bm{\theta} - \hat{\bm{\theta}}) \\
&\quad + \frac{1}{2} (\bm{\theta} - \hat{\bm{\theta}})^\top \nabla^2_{\bm{\theta}}\log p(\hat{\bm{\theta}}) (\bm{\theta} - \hat{\bm{\theta}}),
\end{aligned}
\]
\[
\begin{aligned}
\log q(\bm{\theta}) &= \log q(\hat{\bm{\theta}}) + \nabla_{\bm{\theta}}\log q(\hat{\bm{\theta}})^\top (\bm{\theta} - \hat{\bm{\theta}}) \\
&\quad + \frac{1}{2} (\bm{\theta} - \hat{\bm{\theta}})^\top \nabla^2_{\bm{\theta}}\log q(\hat{\bm{\theta}}) (\bm{\theta} - \hat{\bm{\theta}}),
\end{aligned}
\]
\[
\begin{aligned}
\log q(\bm{z}) &= \log q(\hat{\bm{z}}) + \nabla_{\bm{z}}\log q(\hat{\bm{z}})^\top (\bm{z} - \hat{\bm{z}}) \\
&\quad + \frac{1}{2} (\bm{z} - \hat{\bm{z}})^\top \nabla^2_{\bm{z}}\log q(\hat{\bm{z}}) (\bm{z} - \hat{\bm{z}}).
\end{aligned}
\]

To recover the ELBO, take expectations over these expansions. Notice that the zero-order terms, such as $\log p(\hat{\bm{z}}\mid\hat{\bm{\theta}})$ or $\log p(\hat{\bm{\theta}})$, are constant with respect to $q(\bm{z})$ and $q(\bm{\theta})$ and can thus be moved outside the expectations. Additionally, the first-order terms vanish in expectation.  Thus, we can write
\begin{equation}
\mathbb{E}_{q(\bm{\theta})q(\bm{z})}\left[\log\frac{p(\bm{z}\mid\bm{\theta})p(\bm{\theta})}{q(\bm{\theta})q(\bm{z})}\right] = \log\frac{p(\hat{\bm{z}}\mid\hat{\bm{\theta}})p(\hat{\bm{\theta}})}{q(\hat{\bm{\theta}})q(\hat{\bm{z}})} + C,
\end{equation}
where $C$ collects the second-order terms of the Taylor expansions. Note the first term on the right-hand side is equal to $\log p_{\text{VBC}}(\bm{y})$. We now proceed to show that for the probit case and under a multivariate Gaussian prior $p(\bm{\theta})$, these second-order terms cancel exactly, i.e., $C=0$. 

The second order term with respect to the truncated Gaussian $q(\bm{z})$ can be worked out to be 
\begin{equation}
\frac{1}{2}(\bm{z}-\hat{\bm{z}})^\top \nabla^2_{\bm{z}} \log q(\hat{\bm{z}})(\bm{z}-\hat{\bm{z}}) = -\frac{1}{2}\sum_{i=1}^n (z_i-\hat{z}_i)^2
\end{equation}
and thus
\begin{equation}
    \mathbb{E}_{q(\bm{z})}\left[-\frac{1}{2}\sum_{i=1}^n (z_i-\hat{z}_i)^2\right]
= -\frac{1}{2}\sum_{i=1}^n \mathbb{V}_{q(z_i)}[z_i].
\end{equation}

For the multivariate Gaussian $q(\bm{\theta})$ with mean $\bm{\mu}_{\bm{\theta}}$ and covariance $\bm{V}_{\bm{\theta}}$, where $\bm{\mu}_{\bm{\theta}} = \hat{\bm{\theta}}$, the second order term in the Taylor expansion is
\[
\frac{1}{2} (\bm{\theta}-\hat{\bm{\theta}})^\top \nabla^2_{\bm{\theta}} \log q(\hat{\bm{\theta}}) (\bm{\theta}-\hat{\bm{\theta}})
= -\frac{1}{2} (\bm{\theta}-\hat{\bm{\theta}})^\top \bm{V}_{\bm{\theta}}^{-1} (\bm{\theta}-\hat{\bm{\theta}}),
\]
where $\bm{V}_{\bm{\theta}}$ is the covariance matrix of $q(\bm{\theta})$.

Taking the expectation under $q(\bm{\theta})$ yields
\[
\mathbb{E}_{q(\bm{\theta})}\!\left[-\frac{1}{2} (\bm{\theta}-\hat{\bm{\theta}})^\top \bm{V}_{\bm{\theta}}^{-1} (\bm{\theta}-\hat{\bm{\theta}})\right]
= -\frac{1}{2}\operatorname{tr}\Bigl(\bm{V}_{\bm{\theta}}^{-1}\mathbb{E}_{q(\bm{\theta})}[(\bm{\theta}-\hat{\bm{\theta}})(\bm{\theta}-\hat{\bm{\theta}})^\top]\Bigr)
= -\frac{1}{2}\operatorname{tr}\bigl(\bm{V}_{\bm{\theta}}^{-1}\bm{V}_{\bm{\theta}}\bigr)
= -\frac{d}{2},
\]
where $d=p+1$ is the dimensionality of $\bm{\theta} = (\alpha, \bm{\beta})$. 

For the prior part, expanding $\log p(\bm{\theta})$ around $\hat{\bm{\theta}}$ gives the second-order term
\[
\frac{1}{2} (\bm{\theta}-\hat{\bm{\theta}})^\top \nabla^2_{\bm{\theta}} \log p(\hat{\bm{\theta}}) (\bm{\theta}-\hat{\bm{\theta}})
= -\frac{1}{2} (\bm{\theta}-\hat{\bm{\theta}})^\top \bm{V}_0^{-1} (\bm{\theta}-\hat{\bm{\theta}}),
\]
where $\bm{V}_0$ is the prior covariance matrix. Taking the expectation with respect to $q(\bm{\theta})$, we obtain
\[
\mathbb{E}_{q(\bm{\theta})}\!\left[-\frac{1}{2} (\bm{\theta}-\hat{\bm{\theta}})^\top \bm{V}_0^{-1} (\bm{\theta}-\hat{\bm{\theta}})\right]
= -\frac{1}{2}\operatorname{tr}\Bigl(\bm{V}_0^{-1}\mathbb{E}_{q(\bm{\theta})}\bigl[(\bm{\theta}-\hat{\bm{\theta}})(\bm{\theta}-\hat{\bm{\theta}})^\top\bigr]\Bigr)
= -\frac{1}{2}\operatorname{tr}\bigl(\bm{V}_0^{-1}\bm{V}_{\bm{\theta}}\bigr).
\]

For the $\log p(\bm{z}|\bm{\theta})$ terms, note that the scale parameter of these Gaussians is fixed to $1$ for the probit model. Hence, for each $i$, the second derivative with respect to $z_i$ is
\[
\frac{\partial^2}{\partial z_i^2}\log p(z_i\mid\bm{\theta}) = -1,
\]
so that the Hessian with respect to $\bm{z}$ is
\[
\nabla^2_{\bm{z},\bm{z}}\log p(\bm{z}\mid\bm{\theta}) = -\bm{I}_n.
\]

For $\bm{\theta}$, differentiating twice with respect to $\bm{\theta}$ yields
\[
\nabla^2_{\bm{\theta},\bm{\theta}}\log p(\bm{z}\mid\bm{\theta}) = -\sum_{i=1}^n \bm{W}_i \bm{W}_i^\top
= -\bm{W}^\top \bm{W},
\]
where $\bm{W} = (\bm{\iota}_n, \bm{X})$.

For the cross-derivatives, note that in expectation with respect to $q(\bm{z})q(\bm{\theta})$, the cross terms vanish, e.g.:
\[
\mathbb{E}_{q(\bm{z})q(\bm{\theta})}\Bigl[(\bm{z}-\hat{\bm{z}})^\top \nabla^2_{\bm{z},\bm{\theta}}\log p(\hat{\bm{z}},\hat{\bm{\theta}})(\bm{\theta}-\hat{\bm{\theta}})\Bigr] 
= \Bigl(\mathbb{E}_{q(\bm{z})}[\bm{z}-\hat{\bm{z}}]\Bigr)^\top \nabla^2_{\bm{z},\bm{\theta}}\log p(\hat{\bm{z}},\hat{\bm{\theta}}) \mathbb{E}_{q(\bm{\theta})}[\bm{\theta}-\hat{\bm{\theta}}] = 0.
\]

For the $\bm{z}$-dependent part we have:
\[
\frac{1}{2}(\bm{z}-\hat{\bm{z}})^\top \bigl(-\bm{I}_n\bigr)(\bm{z}-\hat{\bm{z}})
= -\frac{1}{2}\sum_{i=1}^n (z_i-\hat{z}_i)^2.
\]
Taking the expectation with respect to $q(\bm{z})$ gives:
\[
\mathbb{E}_{q(\bm{z})}\left[-\frac{1}{2}\sum_{i=1}^n (z_i-\hat{z}_i)^2\right]
= -\frac{1}{2}\sum_{i=1}^n \mathbb{E}_{q(z_i)}\left[(z_i-\hat{z}_i)^2\right] = -\frac{1}{2}\sum_{i=1}^n \mathbb{V}_{q(z_i)}[z_i].
\]

Next, for the $\bm{\theta}$-dependent part:
\[
\frac{1}{2}(\bm{\theta}-\hat{\bm{\theta}})^\top \bigl(-\bm{W}^\top \bm{W}\bigr)(\bm{\theta}-\hat{\bm{\theta}})
= -\frac{1}{2}(\bm{\theta}-\hat{\bm{\theta}})^\top \bm{W}^\top \bm{W} (\bm{\theta}-\hat{\bm{\theta}}).
\]
Taking the expectation with respect to $q(\bm{\theta})$ yields:
\[
\mathbb{E}_{q(\bm{\theta})}\left[-\frac{1}{2}(\bm{\theta}-\hat{\bm{\theta}})^\top \bm{W}^\top \bm{W} (\bm{\theta}-\hat{\bm{\theta}})\right]
= -\frac{1}{2}\operatorname{tr}\Bigl(\bm{W}^\top \bm{W}\bm{V}_{\bm{\theta}}\Bigr).
\]

Thus, the total expected second order contribution from $\log p(\bm{z}\mid\bm{\theta})$ is:
\[
\begin{aligned}
\mathbb{E}_{q(\bm{z})q(\bm{\theta})}\Biggl[&\frac{1}{2}(\bm{z}-\hat{\bm{z}})^\top \bigl(-\bm{I}_n\bigr)(\bm{z}-\hat{\bm{z}})
+ \frac{1}{2}(\bm{\theta}-\hat{\bm{\theta}})^\top \bigl(-\bm{W}^\top \bm{W}\bigr)(\bm{\theta}-\hat{\bm{\theta}}) \Biggr] \\
&= -\frac{1}{2}\sum_{i=1}^n \mathbb{V}_{q(z_i)}[z_i]
-\frac{1}{2}\operatorname{tr}\Bigl(\bm{W}^\top \bm{W}\bm{V}_{\bm{\theta}}\Bigr).
\end{aligned}
\]

Collecting all quadratic contributions in $C$ then gives:
\begin{equation}
\begin{aligned}
C &= -\frac{1}{2}\sum_{i=1}^n \mathbb{V}_{q(z_i)}[z_i]
-\frac{1}{2}\operatorname{tr}\Bigl(\bm{W}^\top \bm{W}\bm{V}_{\bm{\theta}}\Bigr) -\frac{1}{2}\operatorname{tr}\bigl(\bm{V}_0^{-1}\bm{V}_{\bm{\theta}}\bigr) + \frac{1}{2}\sum_{i=1}^n \mathbb{V}_{q(z_i)}[z_i] + \frac{d}{2}\\
&=-\frac{1}{2}\operatorname{tr}\Bigl(\bm{W}^\top \bm{W}\bm{V}_{\bm{\theta}}\Bigr) -\frac{1}{2}\operatorname{tr}\bigl(\bm{V}_0^{-1}\bm{V}_{\bm{\theta}}\bigr) + \frac{d}{2}\\
&=-\frac{1}{2}\operatorname{tr}\bigl((\bm{W}^\top \bm{W} +\bm{V}_0^{-1})\bm{V}_{\bm{\theta}}\bigr) + \frac{d}{2}.
\end{aligned}
\end{equation}

From the variational update equations, we have that $\bm{V}_{\bm{\theta}} = (\bm{W}^\top \bm{W} +\bm{V}_0^{-1})^{-1}$. Therefore:
\begin{equation}
\begin{aligned}
C &= -\frac{1}{2}\operatorname{tr}\bigl((\bm{W}^\top \bm{W} +\bm{V}_0^{-1})\bm{V}_{\bm{\theta}}\bigr) + \frac{d}{2}\\
&= -\frac{1}{2}\operatorname{tr}\bigl(\bm{I}_d\bigr) + \frac{d}{2} = -\frac{d}{2} + \frac{d}{2} = 0.
\end{aligned}
\end{equation}

\subsection{Proposition 6: Model selection consistency (probit)}
\label{app:coroll2}

    For the probit model, asymptotic model selection consistency of the VBC directly follows from the results of \citet{zhang2024bayesian}, due to equivalence of $\log p_{\text{VBC}}(\bm{y})$ and the ELBO (Proposition~\ref{lem:probit_vbc_elbo}). However, it is instructive to approach model selection consistency from a different perspective as well: For the probit model, note that $p(y_i|z_i)=1$ for any $z_i$ in the support of $q(\bm{z})$. Let $\bm{\theta} = (\alpha, \bm{\beta})$. Due to Gaussian $p(z_i|\bm{\theta})$ and truncated Gaussian $q(z_i)$, it is easy to show that $\sum_i\log p(z_i|\bm{\theta}) - \log q(z_i) = \sum_{y_i=0}\log(1-\Phi(\mu_i)) + \sum_{y_i=1}\log\Phi(\mu_i) = \log p(\bm{y}|\bm{\theta})$ at the evaluation point $(\hat{ \bm{z}},\hat{\bm{\theta}})$, due to both $p(z_i|\bm{\theta})$ (explicitly) and $q(z_i)$ (implicitly) being evaluated using the same variational location parameter $\hat{\mu}_i$. As a result, condition (v) in Proposition~\ref{lem:vbc_consistency} is fulfilled, and the asymptotic relation (\ref{eq:asymptotic_vbc}) collapses to 
\begin{equation}
    \log p_{\text{VBC}}(\bm{y}) \xrightarrow[n\rightarrow\infty]{} \log p(\bm{y}|\hat{\bm{\theta}}) - \frac{p+1}{2}\log(n) + \mathcal{O}(1),
\end{equation}
which is the form of the BIC. Under the usual assumption that the `true' parameter value scores the highest possible likelihood asymptotically, model selection consistency immediately follows.

\subsection{Proposition 7: Posterior existence (Tobit)}\label{app:lemm5}

We extend the posterior existence results of \citet{steel2024model} to the
standard left-censored Tobit type I regression.  The latent regression is  
\[
  z_i \;=\; \alpha + \mathbf x_i^{\!\top}\boldsymbol\beta + \varepsilon_i,
  \qquad
  \varepsilon_i\stackrel{\text{i.i.d.}}{\sim} N(0,\sigma^2),
  \qquad i=1,\dots ,n,
\]
and the observation equation
\[
  y_i \;=\;
  \begin{cases}
    z_i, & z_i>y_L,\\[2pt]
    y_L, & z_i\le y_L,
  \end{cases}
\]
where $y_L$ is a fixed censoring point. For an \emph{uncensored} outcome $y_i>y_L$ the latent value is observed exactly. For a \emph{censored} observation $y_i=y_L$ we only know that $z_i\le y_L$:
\[
  p(y_i\mid z_i)=\mathbb I\!\bigl(z_i\le y_L\bigr).
\]

For every candidate model $M_k$ with design matrix $\mathbf X_k$
we adopt the usual g-prior on $\bm{\beta_k}$ and the improper prior $p(\alpha, \sigma^2)\;\propto\;\sigma^{-2}$ for the latent intercept and error variance. As conditions, we require the matrix $(\bm{\iota}:\mathbf X_k)$ to be of full rank for every model $M_k$ (\textbf{Condition 1}) and at least two uncensored observations $y_i>y_L$ (\textbf{Condition 2}). If and only if these conditions hold, the posterior distribution for every model $M_k$ is proper under the left-censored tobit model and the prior specified above.

\begin{proof}
    
To show this, denote by $\bm{y}$ the vector of all observations $y_{i}$ and partition $\bm{y}$ as $\bm{y}=(\bm{y}_N',\bm{y}_Z')'$ where $\bm{y}_N$ groups all $n_N$ observations that allow for the integral $\int p(y_i|z_i) d z_i$ to be finite. Now consider the marginal likelihood for model $M_k$
\begin{equation}
p(\bm{y}|M_k)=p(\bm{y}_Z|\bm{y}_N, M_k) p(\bm{y}_N| M_k) 
\label{MLy}
\end{equation}
and we need to show that this marginal likelihood is finite for all values of $\bm{y}$ and for any model $M_k$. First, let us focus on the vector $\bm{y}_N$:
\begin{equation}
p(\bm{y}_N|M_k)=\int p(\bm{y}_N|\bm{z}_N, M_k) p(\bm{z}_N| M_k) d \bm{z}_N,\label{MLyN}
\end{equation}
where $\bm{z}_N$ denotes those $z_{i}$ that correspond to $\bm{y}_N$ and we can write
\begin{equation}
p(\bm{y}_N|\bm{z}_N, M_k)=\prod_{i\in\cal N}p({y}_{i}|{z}_{i}),\label{eq:ygivenzN}
\end{equation}
where $\cal N$ is the set of observation indices of $\bm{y}_N$.
Let us now consider $p(\bm {z}_N| M_k)$. If  the matrix $(\bm{\iota}:\bm{X_k})$ is of full column rank ({\bf Condition 1}) and if $n_N \ge 2$ ({\bf Condition 2}), we can derive that

   \begin{equation}
p(\bm{z}_N|M_k)\propto g^{-\frac{p_k}{2}}
|\bm{X}_k'
\bm{X}_k|^\frac{1}{2}
|\bm{A}_k|^{-\frac{1}{2}}
[\bm{z}_N'\bm {P}_k\bm{z}_N]^{-\frac{n_N-1}{2}},
\label{MLzN}
\end{equation}
where \begin{equation}
\bm{P}_k=\bm{I}_{n_N}- (\bm{\iota} : \bm{X}_{k,N})\left( {\begin{array}{cc}
    n_N^{-1}  & \bm{0}' \\
    \bm{0}  &
    \bm{A}_k^{-1} \\
  \end{array} } \right)  
\left(  \begin{array}{c}
\bm{\iota}'  \\
\bm{X}_{k,N}' \\
\end{array} \right)
\end{equation}
and $\bm{A}_k=\bm{X}_{k,N}'\bm{X}_{k,N} + g^{-1} \bm{X}_{k}'\bm{X}_{k}$. Under Condition 1, $\bm{A}_k$ is invertible and for fixed choices of $g$, the expression in (\ref{MLzN}) is almost surely bounded from above by a finite number, say $c$. Therefore, (\ref{MLyN}) becomes
\begin{equation}
p(\bm{y}_N|M_k)< c\prod_{i\in\cal N}\int p({y}_{i}|{z}_{i}) d {z}_{i},
\label{Int}
\end{equation}
and it is sufficient to show that each of the integrals in the above expression is finite. For every $i \in \cal N$ the observation equals the latent
  value, so
  \[
    p(y_i\mid z_i)\;=\;\delta(y_i-z_i).
  \]

where $\delta(\cdot)$ is the Dirac delta function. Hence,
  \[
    \int p(y_i\mid z_i)\,dz_i
      \;=\;\int \delta(y_i-z_i)\,dz_i
      \;=\;1,
    \qquad i\in\cal N.
  \]
  Inserting this into \eqref{Int} gives
  \[
      p(\bm y_N\mid M_k)
      \;<\;c\prod_{i\in\cal N}1
      \;=\;c
      \;<\;\infty,
  \]
  so the marginal likelihood of the uncensored subsample is finite. Thus, we have a well-defined posterior distribution after taking into account at least 2 uncensored observations. These observations in $\bm{y}_N$ will then update the improper prior into a proper posterior, which can then be used as the (proper) prior for the analysis of the zero observations in $\bm{y}_Z$. Of course, the latter will lead to a proper posterior with a finite integrating constant. Thus, $p(\bm{y}_Z| \bm{y}_{N}, M_k)<\infty$ and we obtain that
$p(\bm{y}|M_k)<\infty$,
which proves the result.

\end{proof}

\subsection{Proposition 8: Model selection consistency (tobit)}\label{app:coroll3}

Let $\bm{\theta} = (\alpha, \bm{\beta}, \sigma^2)$. Note that for uncensored observations, $p(\bm{y}|\bm{\theta}) = p(\bm{z}|\bm{\theta})$. For censored observations, $p(y_i|z_i)=1$ for all $z_i$ in the support of $q(\bm{z})$. We will exploit that, as $p(\bm{z}|\bm{\theta})$ is Gaussian and $q(\bm{z})$ is truncated Gaussian, for censored observations the relationship
\begin{equation}
   \log p(z_i|\bm{\theta})-\log q(z_i) =  \log\left[\Phi\left(\frac{y_L-\mu_i}{\sigma}\right)\right] = \log p(y_i|\bm{\theta}),
\end{equation}
holds if both the Gaussian $p(z_i|\bm{\theta})$ and truncated Gaussian $q(z_i)$ are evaluated at the same location and scale parameters. For any censored observation $z_i$, $q(z_i)$ is a truncated Gaussian with location parameter $\mu_i = \mu_{\alpha} + \bm{x}_i^\top \bm{\mu}_{\beta}$ and scale parameter $\xi = \frac{b}{a}$ (Sec.~\ref{sec:tobit}). The Gaussian term $p(z_i|\bm{\theta})$ is evaluated at the posterior expectation $\bm{\hat\theta}$, resulting in the same location parameter, but variance parameter $\frac{b}{a-1}$ (due to the posterior expectation being the expectation of an inverse gamma variational factor for $\sigma^2$). From the definitions of $a$ and $b$ (Sec.~\ref{sec:cavi_updates}) it is easy to see that both $a$ and $b$ grow strictly positively in $n$ and hence, $\frac{b}{a-1} \to \frac{b}{a}$ as $n\to\infty$, and the scale parameters coincide asymptotically. As a result, the left-censored Tobit likelihood 
\begin{equation}
    \log p(\bm{y}|\bm{\theta}) = \sum_{i: y_i > y_L} \left[-\frac{1}{2}\log(2\pi\sigma^2) - \frac{(y_i-\mu_i)^2}{2\sigma^2}\right] + \sum_{i: y_i = y_L} \log \Phi\!\left(\frac{y_L-\mu_i}{\sigma}\right),
\end{equation}
is recovered asymptotically, condition (v) in Proposition~\ref{lem:vbc_consistency} is fulfilled, and the asymptotic relation (\ref{eq:asymptotic_vbc}) collapses to 
\begin{equation}
    \log p_{\text{VBC}}(\bm{y}) \xrightarrow[n\rightarrow\infty]{} \log p(\bm{y}|\hat{\bm{\theta}}) - \frac{p+2}{2}\log(n) + \mathcal{O}(1),
\end{equation}
which is the form of the BIC. Under the usual assumption that the `true' parameter value scores the highest possible likelihood asymptotically, asymptotic selection consistency immediately follows.

\subsection{Proposition 9: Posterior existence (STAR)}\label{app:lemm6}

To extend the posterior existence results of \citet{steel2024model} to the STAR model considered here, note that  

\begin{equation}
    p(y_i|z_i) = \begin{cases}
    \mathbbm{1}(\log y_i < z_i < \log y_i+1)&  \iff y_i>0\\
    \mathbbm{1}(z_i < 0)&  \iff y_i=0\\
\end{cases}
\end{equation}
and hence 
\begin{equation}
    \int p(y_i|z_i) dz_i = \begin{cases}
    \log(y_i+1)-\log(y_i)&  \iff y_i>0\\
    \infty&  \iff y_i=0.\\
\end{cases}
\end{equation}

As a result, we require at least two observations with $0<y_i$ and thus finite $\int p(y_i|z_i)dz_i$ to be present in the data. Posterior existence then follows along the lines of the derivations and under the same full rank condition on $\bm{X}_k$ as given in Sec.~\ref{app:lemm5}.

\subsection{Proposition 10: Model selection consistency (STAR)}
\label{app:coroll4}

Let $\bm{\theta} = (\alpha, \bm{\beta}, \sigma^2)$. For the STAR model, note that $p(y_i|z_i)=\mathbbm{1}(z_i \in \left(a_i, b_i)\right)=1$ for any $z_i$ in the support of $q(\bm{z})$. In addition, if we evaluate the Gaussian $p(z_i|\bm{\theta})$ and truncated Gaussian $q(z_i)$ at the same location and scale parameters, it is easy to show that at the point $\hat{\bm{\theta}}$
\begin{equation}
    \sum_i\log p(z_i|\bm{\theta}) - \log q(z_i) = \sum_{i=1}^n \log\left[\Phi\left(\frac{b_i-\mu_i}{\sigma}\right) - \Phi\left(\frac{a_i-\mu_i}{\sigma}\right)\right]
 = \log p(\bm{y}|\bm{\theta})
\end{equation}
where $a_i$ and $b_i$ are the respective known lower and upper truncation bounds of observation $i$. This is exactly the form of the likelihood given in \citet{kowal2020simultaneous}. From the variational updates in Sec.~\ref{sec:star}, we have that $q(z_i)$ is a truncated Gaussian with location parameter $\mu_i = \mu_{\alpha} + \bm{x}_i^\top \bm{\mu}_{\beta}$ and scale parameter $\xi = \frac{b}{a}$. The Gaussian term $p(z_i|\bm{\theta})$ is evaluated at the posterior expectation $\bm{\hat\theta}$, resulting in the same location parameter, but variance parameter $\frac{b}{a-1}$ (due to the posterior expectation being the expectation of an inverse gamma variational factor for $\sigma^2$). From the definitions of $a$ and $b$ (Sec.~\ref{sec:cavi_updates}) it is easy to see that both $a$ and $b$ grow strictly positively in $n$ and hence, $\frac{b}{a-1} \to \frac{b}{a}$ as $n\to\infty$, and the scale parameters coincide asymptotically. As a result, condition (v) in Proposition~\ref{lem:vbc_consistency} is fulfilled asymptotically, and the asymptotic relation (\ref{eq:asymptotic_vbc}) collapses to 
\begin{equation}
    \log p_{\text{VBC}}(\bm{y}) \xrightarrow[n\rightarrow\infty]{} \log p(\bm{y}|\hat{\bm{\theta}}) - \frac{p+2}{2}\log(n) + \mathcal{O}(1),
\end{equation}
which is the form of the BIC. Under the usual assumption that the `true' parameter value scores the highest possible likelihood asymptotically, asymptotic selection consistency immediately follows.

\subsection{Proposition 11: AVB shrinks VB (univariate case)}
\label{app:avb-shrinkage-1d}
Consider a mean-field variational approximation with Gaussian factor $q_\beta(\beta)=\mathcal N(\mu,\omega)$ where $\beta\in\R$ is scalar, with the remaining factors
$(q_\alpha,\,q_{\sigma^2},\,q_{\bm z})$.
Write $\lambda:=(q_\alpha,q_{\sigma^2})$. Let $\bm{X}\in\R^{n\times 1}$ be the single centered regressor (so $\bm{X}^\top\bm{\iota}=0$ and $\bm{X}^\top \bm{X}>0$). The ELBO is $\mathcal L(\mu,\omega,\lambda,q_{\bm z})$. Define the profiled (supremum) objectives in $\mu$:
\begin{align*}
\text{VB (full profiling):}\quad
\mathcal L^*(\mu)
&:= \sup_{\omega,\ \lambda,\ q_{\bm z}}\ \mathcal L(\mu,\omega,\lambda,q_{\bm z}),\\
\text{AVB (fix $q_{\bm z}$ at null):}\quad
\bar{\mathcal L}(\mu)
&:= \sup_{\omega,\ \lambda}\ \mathcal L(\mu,\omega,\lambda,\bar q_{\bm z}),
\end{align*}
where $\bar q_{\bm z}$ is the optimizer of $q_{\bm z}$ from the \emph{full} VB problem at $\mu=0$ (the null fit). Since VB optimizes over a superset of variables, for all $\mu$, we have $\bar{\mathcal L}(\mu)\le \mathcal L^*(\mu)$ and $\bar{\mathcal L}(0)=\mathcal L^*(0)$ (by construction of $\bar q_{\bm z}$). Assume $\mathcal L^*$ and $\bar{\mathcal L}$ are $C^2$ and strictly concave in $\mu$ and define the nonnegative \emph{gap} $D(\mu):=\mathcal L^*(\mu)-\bar{\mathcal L}(\mu)\ge 0$ with $D(0)=0$. Assume $D$ is convex and has a unique global minimum at $0$. Then, $D'(\mu)$ has the sign of $\mu$ for $\mu\neq 0$. 

\begin{proposition}[AVB shrinkage -- univariate case]
\label{prop:avb-shrinkage-1d}
Let $\hat\beta_{\mathrm{VB}}:=\arg\max_{\mu}\mathcal L^*(\mu)$ and
$\hat\beta_{\mathrm{AVB}}:=\arg\max_{\mu}\bar{\mathcal L}(\mu)$.
Under the conditions above,
\[
\operatorname{sign}(\hat\beta_{\mathrm{AVB}})=\operatorname{sign}(\hat\beta_{\mathrm{VB}})\quad\text{and}\quad
|\hat\beta_{\mathrm{AVB}}|\le |\hat\beta_{\mathrm{VB}}|,
\]
with equality iff $\hat\beta_{\mathrm{VB}}=0$.
\end{proposition}

\begin{proof}
Since $D(0)=0$ and $D\ge 0$ with a unique minimum at $0$, we have $D'(0)=0$ and $D'(\mu)$ has the sign of $\mu$ for $\mu\neq 0$. Because $D'=\mathcal L^{*'}-\bar{\mathcal L}'$, this yields
\[
\bar{\mathcal L}'(0)=\mathcal L^{*'}(0),\qquad
\bar{\mathcal L}'(\mu)=\mathcal L^{*'}(\mu)-D'(\mu).
\]
At $\mu=\hat\beta_{\mathrm{AVB}}$, $\bar{\mathcal L}'(\hat\beta_{\mathrm{AVB}})=0$, hence
$\mathcal L^{*'}(\hat\beta_{\mathrm{AVB}})=D'(\hat\beta_{\mathrm{AVB}})$ has the sign of $\hat\beta_{\mathrm{AVB}}$.
Since $\mathcal L^*$ is strictly concave, its derivative is strictly decreasing and changes sign from $+$ to $-$ at the unique maximizer $\hat\beta_{\mathrm{VB}}$. Therefore, if $\hat\beta_{\mathrm{AVB}}>0$, then $\mathcal L^{*'}(\hat\beta_{\mathrm{AVB}})>0$ implies $0<\hat\beta_{\mathrm{AVB}}<\hat\beta_{\mathrm{VB}}$; similarly, if $\hat\beta_{\mathrm{AVB}}<0$, then $\mathcal L^{*'}(\hat\beta_{\mathrm{AVB}})<0$ implies $\hat\beta_{\mathrm{VB}}<\hat\beta_{\mathrm{AVB}}<0$. If $\hat\beta_{\mathrm{AVB}}=0$, strict concavity gives $\hat\beta_{\mathrm{VB}}=0$. The claims follow.
\end{proof}

\section{Algorithmic Implementation Details}
\label{app:algorithmic}

The following algorithms implement model space exploration using the VBC criterion. For both algorithms, we use an add-delete-swap proposal mechanism where proposal probabilities $r(\mathcal{M}_j|\mathcal{M}_k)$ are computed based on the specific moves between models. Details on these proposal probabilities can be found in the supplement of \citet{steel2024model}. By default, we consider the algorithm converged if the largest relative parameter change from the previous iteration is at most $10^{-6}$, or if the algorithm has performed 10,000 iterations.

\begin{algorithm}[h!]
\setstretch{1.35}
\caption{Model Space Exploration using Variational Bayes (VB) with VBC}
\label{alg:vb_exploration}
\begin{algorithmic}[1]
\setlength{\itemsep}{0pt}
\setlength{\parsep}{0pt}
\Require Data $\bm{y}$, design matrix $\bm{X}$, hyperparameter $g$ (default: $g=n$)
\Require Initial model $\mathcal{M}_k$, number of iterations $T$ (default: $T=10,000$ after burn-in of 2,000)
\State Compute initial approximation $p_{\text{VBC}}(\bm{y}|\mathcal{M}_k)$ using full VB
\For{$t = 1$ to $T$}
    \State Propose new model $\mathcal{M}_j$ using add-delete-swap proposal from $\mathcal{M}_k$
    \State \textbf{Run full VB algorithm for model $\mathcal{M}_j$:}
    \State \quad Initialize $q_j(\alpha)$, $q_j(\bm{\beta})$, $q_j(\sigma^2)$, $q_j(\bm{z})$
    \Repeat
        \State \quad Update $q_j(\alpha)$: $q_j(\alpha) \propto \mathcal{N}(\mu_{\alpha,j}, \omega_{\alpha,j})$
        \State \quad Update $q_j(\bm{\beta})$: $q_j(\bm{\beta}) \propto \mathcal{N}(\bm{\mu}_{\beta,j}, \bm{\Omega}_{\beta,j})$ 
        \State \quad Update $q_j(\sigma^2)$: $q_j(\sigma^2) \propto \mathcal{IG}(a_j, b_j)$ (if not fixed)
        \State \quad Update $q_j(z_i)$ for each $i$ using model-specific updates
    \Until{Convergence}
    \State Compute approximation $p_{\text{VBC}}(\bm{y}|\mathcal{M}_j)$ using equation (3)
    \State Calculate acceptance probability: $\alpha = \min\left\{ 1, \frac{p_{\text{VBC}}(\bm{y}|\mathcal{M}_j)\,p(\mathcal{M}_j)\,r(\mathcal{M}_k|\mathcal{M}_j)}{p_{\text{VBC}}(\bm{y}|\mathcal{M}_k)\,p(\mathcal{M}_k)\,r(\mathcal{M}_j|\mathcal{M}_k)} \right\}$
    \State Generate $u \sim \text{Uniform}(0,1)$
    \If{$u < \alpha$}
        \State Accept: $\mathcal{M}_k \leftarrow \mathcal{M}_j$, $p_{\text{VBC}}(\bm{y}|\mathcal{M}_k) \leftarrow p_{\text{VBC}}(\bm{y}|\mathcal{M}_j)$
    \EndIf
    \State Store current model $\mathcal{M}_k$ and approximation $p_{\text{VBC}}(\bm{y}|\mathcal{M}_k)$
\EndFor
\State \Return Sequence of visited models, parameters and marginal likelihood approximations
\end{algorithmic}
\end{algorithm}

\begin{algorithm}[h!]
\setstretch{1.35}
\caption{Model Space Exploration using Approximate Variational Bayes (AVB) with VBC}
\label{alg:avb_exploration}
\begin{algorithmic}[1]
\Require Data $\bm{y}$, design matrix $\bm{X}$, hyperparameter $g$ (default: $g=n$)
\Require Initial model $\mathcal{M}_k$, number of iterations $T$ (default: $T=10,000$ after burn-in of 2,000)
\State \textbf{Precomputation phase:}
\State Run full VB algorithm for null model $\mathcal{M}_{\text{null}}$ (intercept only)
\State Fix latent variable densities: $\bar{q}(z_i) \leftarrow q_{\text{null}}(z_i)$ from null model for all $i$
\State Compute fixed moments: $\bar{m}_i \leftarrow \mathbb{E}_{\bar{q}(z_i)}[z_i]$, $\bar{s}_i \leftarrow \mathbb{V}_{\bar{q}(z_i)}[z_i]$
\State Compute initial approximation $p_{\text{VBC}}(\bm{y}|\mathcal{M}_k)$ using AVB estimates
\For{$t = 1$ to $T$}
    \State Propose new model $\mathcal{M}_j$ using add-delete-swap proposal from $\mathcal{M}_k$
    \State \textbf{Run AVB algorithm for model $\mathcal{M}_j$:}
    \State \quad Initialize $q_j(\alpha)$, $q_j(\bm{\beta})$, $q_j(\sigma^2)$; Fix $q_j(z_i) \leftarrow \bar{q}(z_i)$ for all $i$
    \Repeat
        \State \quad Update $q_j(\alpha)$ using fixed $\bar{m}_i$: $\mu_\alpha = \frac{1}{n}\sum_{i=1}^n \bar{m}_i$
        \State \quad Update $q_j(\bm{\beta})$: $\bm{\mu}_\beta = \frac{g}{1+g}(\bm{X}_j^\top \bm{X}_j)^{-1}\bm{X}_j^\top \bar{\bm{m}}$
        \State \quad Update $q_j(\sigma^2)$ using fixed $\bar{m}_i$, $\bar{s}_i$ (if not fixed)
    \Until{Convergence}
    \State Compute approximation $p_{\text{VBC}}(\bm{y}|\mathcal{M}_j)$ using equation (3) with fixed $\bar{q}(z_i)$
    \State Calculate acceptance probability: $\alpha = \min\left\{ 1, \frac{p_{\text{VBC}}(\bm{y}|\mathcal{M}_j)\,p(\mathcal{M}_j)\,r(\mathcal{M}_k|\mathcal{M}_j)}{p_{\text{VBC}}(\bm{y}|\mathcal{M}_k)\,p(\mathcal{M}_k)\,r(\mathcal{M}_j|\mathcal{M}_k)} \right\}$
    \State Generate $u \sim \text{Uniform}(0,1)$
    \If{$u < \alpha$}
        \State Accept: $\mathcal{M}_k \leftarrow \mathcal{M}_j$, $p_{\text{VBC}}(\bm{y}|\mathcal{M}_k) \leftarrow p_{\text{VBC}}(\bm{y}|\mathcal{M}_j)$
    \EndIf
    \State Store current model $\mathcal{M}_k$ and approximation $p_{\text{VBC}}(\bm{y}|\mathcal{M}_k)$
\EndFor
\State \Return Sequence of visited models, parameters and marginal likelihood approximations
\end{algorithmic}
\end{algorithm}

\clearpage

\section{Additional Details on Illustrative Models}
\label{sec:additional_details}
\subsection{Master ELBO for Latent Gaussian Regression}
\label{app:common_setup}

Consider latent Gaussian regression with observations $\{(y_i,\bm x_i)\}_{i=1}^n$,
latent variables $\bm z=(z_1,\dots,z_n)^\top$, error variance $\sigma^2$ and linear predictor $\eta_i$
\[
z_i \mid \alpha,\bm\beta,\sigma^2 \ \sim\ \mathcal N\!\big(\eta_i,\ \sigma^2\big),
\qquad
\eta_i := \alpha + \bm x_i^\top\bm\beta.
\]
Let $\bm X$ be the (centered) $n\times p$ design matrix and define $\bm G := \bm X^\top \bm X$. We consider a mean-field variational family
\[
q(\alpha)=\mathcal N(\mu_\alpha,\omega_\alpha),\quad
q(\bm\beta)=\mathcal N(\bm\mu_\beta,\bm\Omega_\beta),\quad
q(\sigma^2)=\mathcal{IG}(a,b)\,,\quad
q(\bm z)=\prod_{i=1}^n q(z_i),
\]
and a $g$-prior $p(\bm\beta\mid\sigma^2)=\mathcal N\!\big(\bm 0,\ g\,\sigma^2\,\bm G^{-1}\big)$ as well as a Jeffreys-type prior $p(\alpha, \sigma^2)\propto\sigma^{-2}$.
Define the variational linear predictor
\[
\mu_i := \mu_\alpha + \bm x_i^\top\bm\mu_\beta,
\qquad
\bm\mu := (\mu_1,\dots,\mu_n)^\top,
\]
and the shared quantities
\[
\tau := \E_q[\sigma^{-2}] = a/b,\qquad
\xi := \frac{b}{a}=\tau^{-1},\qquad
C_\beta := \operatorname{tr}(\bm G\,\bm\Omega_\beta)+\bm\mu_\beta^\top\bm G\,\bm\mu_\beta, \qquad \E_q[\log\sigma^2]=\log b-\psi(a).
\]

\noindent Furthermore, for models involving truncated Gaussian variational factors, for a truncation interval $[\tau_i^{\text{low}},\tau_i^{\text{upp}})$, set
\[
\ell_i=\frac{\tau_i^{\text{low}}-\mu_i}{\sqrt{\xi}},\qquad
u_i=\frac{\tau_i^{\text{upp}}-\mu_i}{\sqrt{\xi}},\qquad
\kappa_i=\Phi(u_i)-\Phi(\ell_i),
\]
and
\[
\lambda_i:=\frac{\phi(\ell_i)-\phi(u_i)}{\kappa_i},\qquad
\chi_i:=1+\frac{\ell_i\phi(\ell_i)-u_i\phi(u_i)}{\kappa_i}-\lambda_i^2.
\]
If $q(z_i)$ is $\mathcal N(\mu_i,\xi)$ truncated to $[\tau_i^{\text{low}},\tau_i^{\text{upp}})$, then by the properties of truncated Gaussians, we have
\[
m_i:=\mathbb{E}_{q(z_i)}[z_i]=\mu_i+\sqrt{\xi}\,\lambda_i,
\qquad
s_i:=\mathbb{V}_{q(z_i)}[z_i]=\xi\,\chi_i.
\]
\noindent Finally, from the properties of the Gaussian and inverse gamma variational factors for $\alpha$, $\bm{\beta}$ and $\sigma^2$, the following entropy terms can be derived:
\begin{align}    
\mathcal H_{\alpha,\beta}
=\tfrac12\log(2\pi e\,\omega_\alpha)
+\tfrac12\log\!\bigl((2\pi e)^p|\bm\Omega_\beta|\bigr),
\quad
\mathcal H_{\sigma^2}
=a+\log b+\log\Gamma(a)-(a+1)\psi(a).
\end{align}

For the variational densities $q(z_i)$, which are Gaussian or truncated Gaussian in our setting, we have the following entropy terms:
\[
\mathcal H_{q(z_i)}=
\begin{cases}
\displaystyle \tfrac12\log(2\pi e\,s_i), & \text{if } q(z_i)=\mathcal N(m_i,s_i),\\[6pt]
\displaystyle \log\!\bigl(\kappa_i\sqrt{2\pi e\,\xi}\bigr)
+\dfrac{\ell_i\phi(\ell_i)-u_i\phi(u_i)}{2\kappa_i}, & \text{if $q(z_i)=\mathcal{TN}(\cdot)$.}
\end{cases}
\]

Then, a \emph{master ELBO for latent Gaussian regression} is given by:
\begin{equation}
\label{eq:master_elbo}
\begin{aligned}
\mathrm{ELBO}
&= \sum_{i=1}^n \E_q[\log p(y_i\mid z_i)]
- \frac{n}{2}\log(2\pi) - \frac{p}{2}\log(2\pi)
- \frac{n+p}{2}\,\E_q[\log\sigma^2]
\\
&\quad - \frac{\tau}{2}\sum_{i=1}^n\Big[s_i+\omega_\alpha+\bm x_i^\top\bm\Omega_\beta\bm x_i+\big(m_i-\mu_i\big)^2\Big]
\\
&\quad - \frac{p}{2}\log g + \frac{1}{2}\log|\bm G|
- \frac{\tau}{2g}\,C_\beta
+ \underbrace{\big(\psi(a)-\log b\big)}_{\E_q[\log p(\sigma^2)]}
\\
&\quad + \sum_{i=1}^n \mathcal H_{q(z_i)} + \mathcal H_{\alpha,\beta} + \mathcal H_{\sigma^2}.
\end{aligned}
\end{equation}

If the model has no $\sigma^2$ (e.g., probit with unit latent variance), set $\tau=1$, $\xi=1$ and drop the ELBO terms related to $\sigma^2$.

\subsection{Probit Model}
\label{app:probit_details}

The probit likelihood is deterministic in $z_i$:
\[
p(y_i\mid z_i)=\mathbbm 1\{z_i>0\}^{\,y_i}\,\mathbbm 1\{z_i\le 0\}^{\,1-y_i}
\ \Rightarrow\ 
\E_q[\log p(y_i\mid z_i)]=0.
\]
The latent prior variance $\sigma^2=1$ is fixed, hence set $\tau=1$, $\xi=1$. By the CAVI update, $q(z_i)$ is truncated to the observation-consistent halfline
\[
(\tau_i^{\text{low}},\tau_i^{\text{upp}})=
\begin{cases}
(0,\infty) & y_i=1,\\
(-\infty,0] & y_i=0,
\end{cases}
\qquad
\ell_i=\begin{cases}-\mu_i & y_i=1,\\-\infty & y_i=0,\end{cases}\qquad u_i=\begin{cases}+\infty & y_i=1,\\-\mu_i & y_i=0.\end{cases}
\]
and the truncated Gaussian moments specialize to
\[
m_i
=\mu_i+\frac{\phi(\mu_i)}{\Phi(\mu_i)}\ \ (y_i=1),
\qquad
m_i
=\mu_i-\frac{\phi(\mu_i)}{1-\Phi(\mu_i)}\ \ (y_i=0),
\]
and $s_i$ follows from the truncated-variance formula with $\xi=1$. The $z$-entropy uses the truncated-Gaussian expression. To derive the probit ELBO, insert the quantities above into \eqref{eq:master_elbo} with $\tau=1$ and drop terms related to $\sigma^2$.

\subsection{Tobit Model (Left-Censored at $y_L$)}
\label{app:tobit_details}

The tobit likelihood is an indicator consistent with the support, hence $\E_q[\log p(y_i\mid z_i)]=0$. For censored observations ($y_i=y_L$), truncate to $(-\infty,y_L]$:
\[
\ell_i=-\infty,\qquad u_i=\eta_i:=\frac{y_L-\mu_i}{\sqrt{\xi}},
\]
\[
m_i=\mu_i-\sqrt{\xi}\,\frac{\phi(\eta_i)}{\Phi(\eta_i)},\qquad
s_i=\xi\Big[1-\eta_i\frac{\phi(\eta_i)}{\Phi(\eta_i)}-\Big(\frac{\phi(\eta_i)}{\Phi(\eta_i)}\Big)^2\Big].
\]
For uncensored observations ($y_i>y_L$), treat $z_i$ as observed:
\[
m_i=y_i,\qquad s_i=0,\qquad \mathcal H_{q(z_i)}=0.
\]
For censored $i$, use the truncated-entropy formula; for uncensored $i$, no $z$-entropy term enters. Plugging into \eqref{eq:master_elbo} gives the tobit ELBO term. Note that only censored indices contribute nonzero $z$-entropy. 

\subsection{STAR Model}
\label{app:star_details}

The STAR likelihood we consider is an interval indicator, thus $\E_q[\log p(y_i\mid z_i)]=0$. All observations use truncated $q(z_i)$ with observation-specific intervals, leading to
\[
[\tau_i^{\text{low}},\tau_i^{\text{upp}})=
\begin{cases}
(-\infty,0) & y_i=0,\\
[\log y_i,\ \log(y_i+1)) & y_i>0,
\end{cases}
\]
\[
m_i=\mu_i+\sqrt{\xi}\,\frac{\phi(\ell_i)-\phi(u_i)}{\kappa_i},\qquad
s_i=\xi\left[1+\frac{\ell_i\phi(\ell_i)-u_i\phi(u_i)}{\kappa_i}
-\left(\frac{\phi(\ell_i)-\phi(u_i)}{\kappa_i}\right)^{2}\right].
\]
The above quantities can be plugged into \eqref{eq:master_elbo} and the ELBO term for STAR models follows.

\subsection{Poisson Log-Normal Model (PLN)}
\label{app:pln_details}

For PLN models, we have
\[
\E_q[\log p(y_i\mid z_i)]
= \E_q\big[y_i z_i - e^{z_i} - \log(y_i!)\big]
= y_i m_i - \exp\!\big(m_i+\tfrac12 s_i\big) - \log(y_i!).
\]
Choosing Gaussian $q(z_i)=\mathcal N(m_i,s_i)$ leads to the entropy terms $\mathcal H_{q(z_i)}=\tfrac12\log(2\pi e\,s_i)$. Substituting the above likelihood block and entropy into \eqref{eq:master_elbo} leads to the PLN ELBO expression. Making the variational parameters explicit, the contribution of observation $i$ can be written as
\[
\text{ELBO}_i = y_i m_i - \exp(m_i + 0.5 s_i) - \frac{a}{2b} (m_i - \mu_{\alpha} - \bm{x}_i^\top \bm{\mu}_{\beta})^2 - s_i \frac{a}{2b} + \frac{1}{2} \log s_i + C.
\]

To optimize the variational parameters $m_i$ and $s_i$, we employ parallel Newton-Raphson algorithms for each $i$, utilizing the first and second partial derivatives of the ELBO with respect to $m_i$ and $s_i$. These derivatives are available in closed form:

\begin{align}
\frac{\partial \text{ELBO}_i}{\partial m_i} &= y_i - \exp(m_i + 0.5 s_i) - \frac{a}{b} \left( m_i - \mu_\alpha - \bm{x}_i^\top \bm{\mu}_\beta \right), \\[8pt]
\frac{\partial^2 \text{ELBO}_i}{\partial m_i^2} &= -\exp(m_i + 0.5 s_i) - \frac{a}{b}, \\[8pt]
\frac{\partial \text{ELBO}_i}{\partial s_i} &= -0.5 \exp(m_i + 0.5 s_i) - \frac{a}{2b} + \frac{1}{2 s_i}, \\[8pt]
\frac{\partial^2 \text{ELBO}_i}{\partial s_i^2} &= -0.25 \exp(m_i + 0.5 s_i) - \frac{1}{2 s_i^2}.
\end{align}

As the ELBO factorizes over $i$, these gradient-based updates can be conducted in parallel for each observation. Note only the first part of this gradient-based update depends on the likelihood $p(y_i|z_i)$ term. Hence, much of these derivations are generic and extend to stochastic links beyond the Poisson case.

Note that we did not use a Laplace approximation for $q(z_i)$, which is an alternative approach. Empirically, however, assuming $q(z_i) = \mathcal{N}(m_i, s_i)$ and directly optimizing the variational parameters tends to yield a more accurate variational approximation compared to a Gaussian approximation centered around the mode. This is due to the potential significant skewness in the posterior of $z_i$ for small counts, where a Gaussian approximation expanded around the mode is not very accurate.

\begin{figure}[!t]
    \centering
    \includegraphics[width=\linewidth]{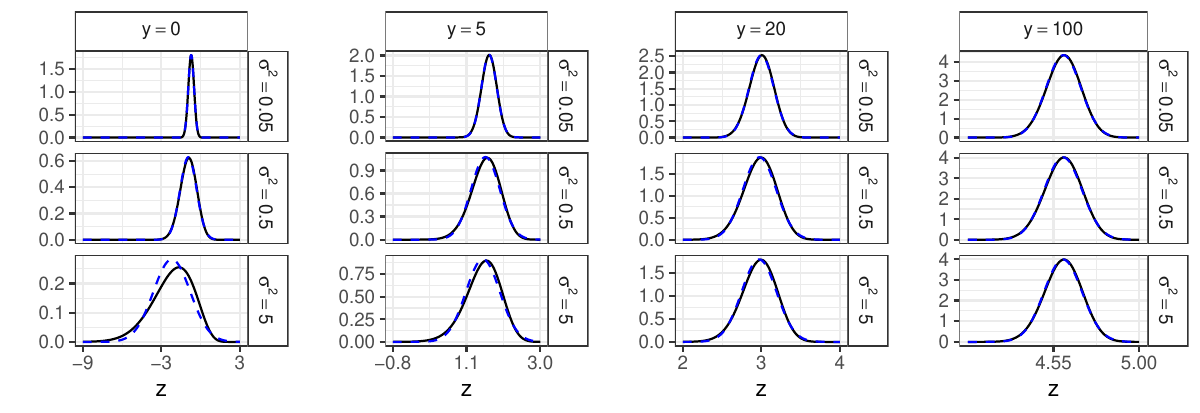}
    \caption{True conditional of $z_i$ (black) and variational approximation (blue dashed) for various settings under a Poisson log-normal model. Columns correspond to size of the underlying count $y_i$. Priors are Gaussian and range from very strong with $\sigma^2=0.05$ to weak with $\sigma^2=5$. All priors are centered at $\log(y+0.5)$, close to the mode of the count likelihood.}
    \label{fig:z_approx}
\end{figure}

Fig.~\ref{fig:z_approx} compares the true conditionals $p(z_i|\cdot)$ with the variational approximations. In line with theoretical expectations, we see that the `worst case' is a small count ($y_i$ close to zero) under a weak Gaussian prior $z_i \sim \mathcal{N}(\alpha + \bm{x}_i^\top\bm{\beta}, \sigma^2)$. In this case, the conditional posterior is skewed, and the variational Gaussian approximation is not able to capture this behavior. However, in virtually all other cases ($y_i$ reasonably large and/or an informative Gaussian prior), a variational Gaussian density is a good description of the true posterior.

\section{Additional Simulation Results}
\label{app:additional_sim}

\subsection{Additional Simulation Settings}

To provide an additional comparison of the AVB and VB algorithms' performance and illustrate the mechanisms underlying the behavior observed in our Afrobarometer real data analysis (Section~\ref{sec:applications}), we conduct the following simulation study. We generate data from a probit model with a true coefficient vector $\bm{\beta} = (0.5, -0.5, 0.25, -0.25, 0.1, \ldots, 0.1)$ and $p=10$ covariates. While the full model represents the true specification, only four covariates exhibit strong relevance, with the remaining six having modest effect sizes.

We vary the sample size as $n \in \{100, 500, 1000, 5000, 10000\}$ and examine two covariate structures to assess likelihood informativeness. In the `independent' setting, predictors are drawn i.i.d.\ from $\mathcal{N}(0,1)$. In the `correlated' setting, observations $\bm{x}_i$ follow a multivariate normal distribution with mean zero and covariance matrix $\bm{\Sigma}$, where $\Sigma_{jk}=\rho^{|j - k|}$ with $\rho = 0.6$. Compared to $\rho = 0.25$ in Section~\ref{sec:numerical}, this higher correlation creates a more challenging setting that emulates datasets with substantial collinearity among predictors.

Furthermore, we investigate the balancedness of the outcome as a second important determinant of likelihood information content in binary settings (\citealp{zens2024ultimate}). We consider two scenarios: a `balanced' case with intercept $\alpha=0$, resulting in approximately 50\% success probability on average, and an `imbalanced' case with $\alpha=-2$, where the negative intercept induces a high proportion of observations with $y_i=0$. These scenarios are deliberately chosen to examine model selection behavior when the amount of likelihood information varies. Scenarios with balanced outcomes and uncorrelated regressors provide more information per observation compared to scenarios with imbalanced outcomes and highly correlated regressors.

We perform full model space enumeration using both AVB and VB, with 40 replications per configuration. Marginal likelihood approximations are based on the VBC. Fig.~\ref{fig:add_sim_res} presents the relationship between posterior inclusion probabilities estimated by the AVB and VB algorithms, averaged across the 40 replications.

In settings with informative likelihood (balanced outcomes and uncorrelated regressors), the asymptotic equivalence of AVB and VB PIPs emerges rapidly as sample size increases. Both algorithms successfully identify the true model with high confidence when $n\geq5,000$. However, in more challenging scenarios characterized by imbalanced outcomes and correlated regressors, differences between the algorithms persist even with $n=10,000$ observations. Here, the AVB algorithm fails to recover the correct median probability model, whereas the VB algorithm achieves accurate model identification. Importantly, as sample size increases, both algorithms progressively shift all PIPs toward unity, consistent with the true underlying model structure. As shown in Section~\ref{sec:numerical}, both algorithms eventually recover the correct model as $n$ grows sufficiently large.

A key observation is that algorithmic disagreement occurs exclusively in the upper-left quadrant of the plots (indicated by green triangles), where AVB excludes variables that VB includes in the median probability model. This pattern underscores the conservative nature of the AVB algorithm, as predicted by Proposition~\ref{prop:local-shrink-simple}. These empirical results underscore that AVB typically requires more likelihood information than VB for reliable model selection and model averaging performance. This additional simulation study provides important context for interpreting the differential algorithm performance observed in the Afrobarometer application (Section~\ref{sec:applications}), where the combination of binary outcomes, correlated predictors, and large model space creates precisely the type of challenging scenario examined here.

\begin{figure}[t]
    \centering
    \includegraphics[width=0.9\textwidth]{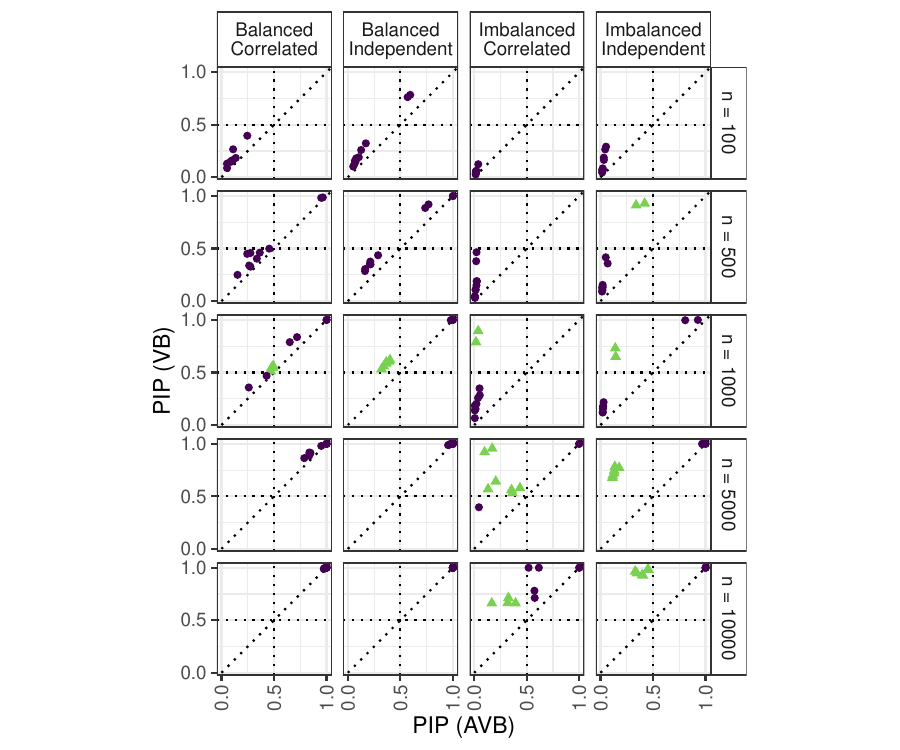}
    \caption{Posterior inclusion probabilities from probit AVB (horizontal axis) versus VB (vertical axis) algorithms, computed via full model space enumeration across multiple simulation scenarios. Each point represents the average PIP across 40 replicated datasets. Green triangles indicate variables where the median probability models differ between algorithms, with AVB excluding variables that VB includes. The data generating process uses $\bm{\beta} = (0.5, -0.5, 0.25, -0.25, 0.1, \dots, 0.1)^\top$ with $p=10$ covariates. Sample sizes range from $n=100$ to $n=10,000$. `Balanced' refers to $\alpha=0$ (approximately 50\% success outcomes), while `imbalanced' uses $\alpha=-2$ (predominantly zero outcomes). `Correlated' regressors have covariance $\Sigma_{jk}=0.6^{|j-k|}$, while `independent' regressors are i.i.d.\ $\mathcal{N}(0,1)$.}
    \label{fig:add_sim_res}
\end{figure}

\clearpage
\subsection{Additional Results}

\begin{figure}[ht]
    \centering
    \includegraphics[width=0.7\textwidth]{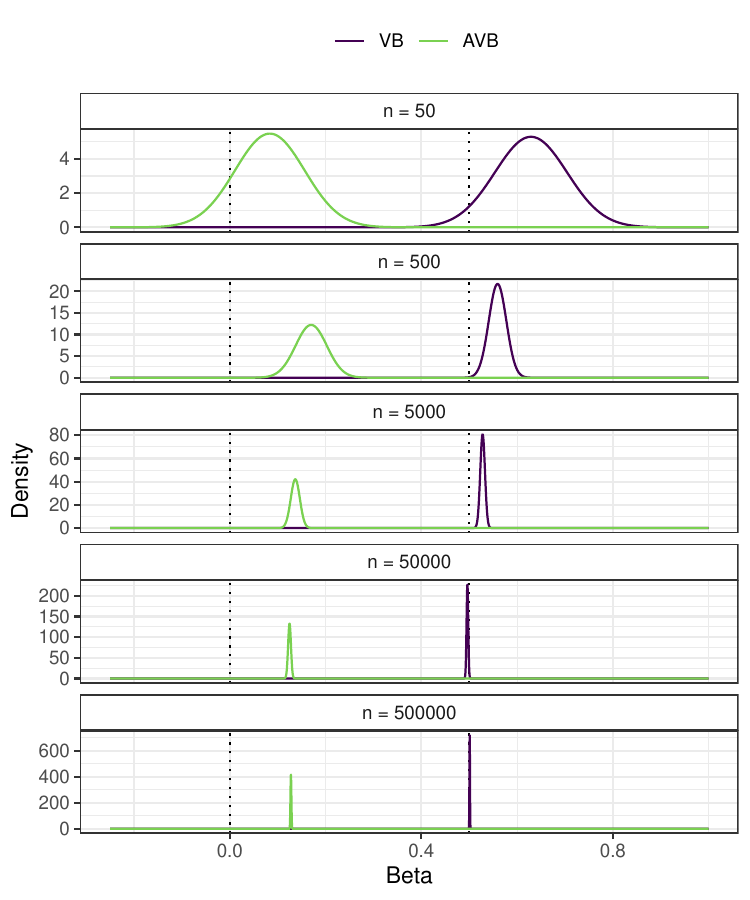}
    \caption{Posterior densities of $\beta_1$ under AVB (green) and VB (purple) algorithms for sample sizes ranging from $n=50$ to $n=500,000$. The data generating process is a Poisson log-normal model with $p=10$ covariates, $\alpha=0$, $\sigma^2=0.1$, and true parameter vector $\bm{\beta} = (0.5,-0.5, 0.25,-0.25, 0, \dots, 0)$. $\beta_1$ is shown on the horizontal axis. Vertical dashed lines indicate zero (left) and the true parameter value $\beta_1=0.5$ (right). The figure illustrates asymptotic posterior consistency of the full VB algorithm and shows that AVB estimates tend to concentrate at a point between 0 and the full VB estimate, consistent with Proposition~\ref{prop:local-shrink-simple}.}
    \label{fig:avb_convergence}
\end{figure}

\begin{figure}[ht]
    \centering
    \includegraphics[width=\textwidth]{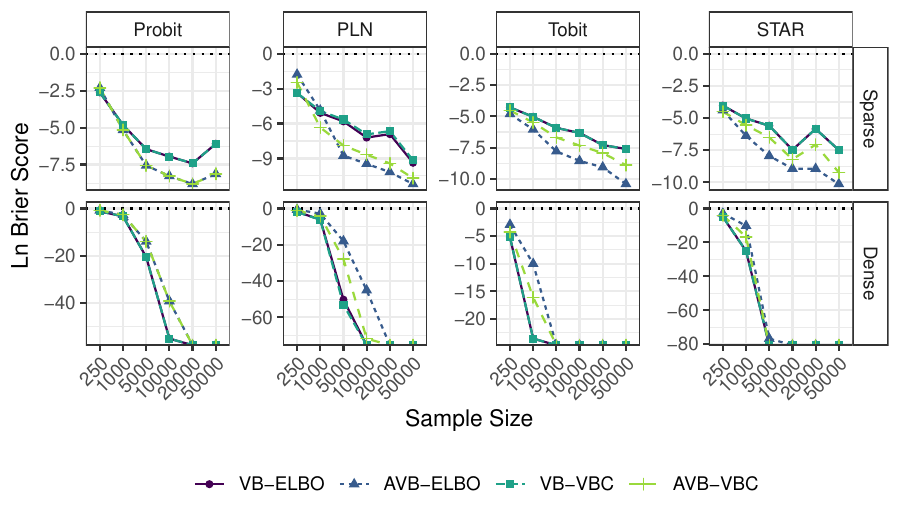}
    \caption{Log scale Brier scores for posterior inclusion probability estimation in the $p=10$ setting with full model space enumeration. Lower values indicate better performance. Comparison includes VBC criterion with full VB (VB-VBC) and approximate variational Bayes (AVB-VBC), and ELBO criterion with full VB (VB-ELBO) and AVB (AVB-ELBO). Results averaged across 40 replicate datasets. Sample sizes range from $n=250$ to $n=50,000$.}
    \label{fig:brier_enumeration_log}
\end{figure}

\begin{figure}[ht]
    \centering
    \includegraphics[width=\textwidth]{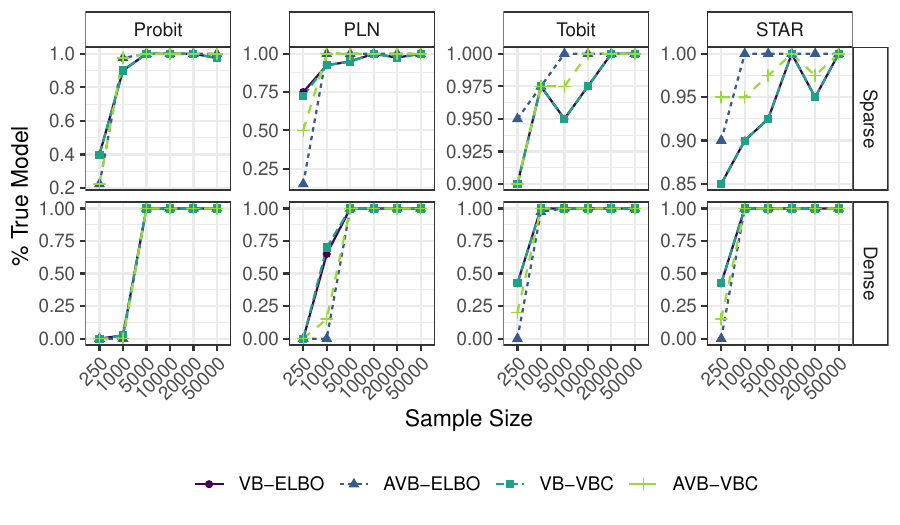}
    \caption{Proportion of simulation runs in which the true model was selected as the highest-scoring model in the $p=10$ setting with full model space enumeration (1,024 models). Comparison of model selection criteria: VBC with full VB (VB-VBC) and AVB (AVB-VBC), ELBO with full VB (VB-ELBO) and AVB (AVB-ELBO). Results based on 40 replicate datasets per configuration. Sample sizes range from $n=250$ to $n=50,000$.}
    \label{fig:true_enumeration}
\end{figure}

\begin{figure}[ht]
    \centering
    \includegraphics[width=\textwidth]{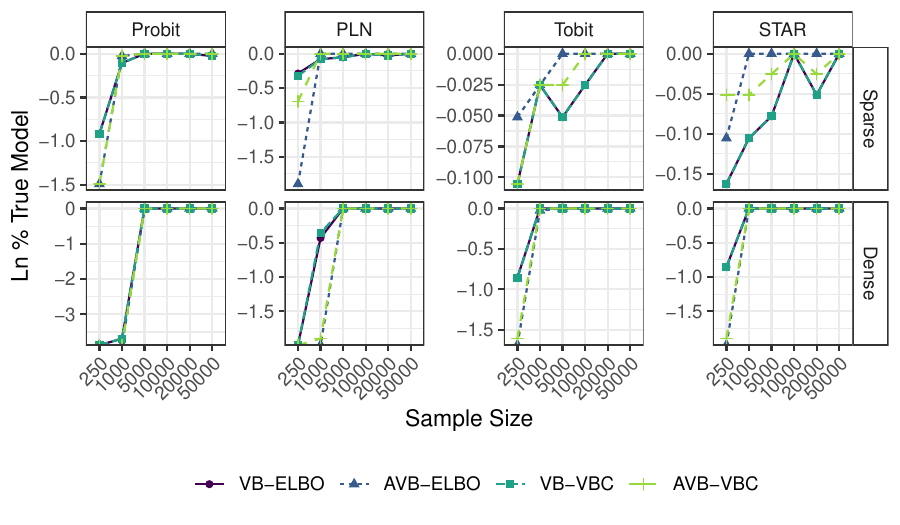}
    \caption{Log scale proportion of simulation runs in which the true model was selected as the highest-scoring model in the $p=10$ setting with full model space enumeration (1,024 models). Comparison of model selection criteria: VBC with full VB (VB-VBC) and AVB (AVB-VBC), ELBO with full VB (VB-ELBO) and AVB (AVB-ELBO). Results based on 40 replicate datasets per configuration. Sample sizes range from $n=250$ to $n=50,000$.}
    \label{fig:true_enumeration_log}
\end{figure}

\begin{figure}[ht]
    \centering
    \includegraphics[width=\textwidth]{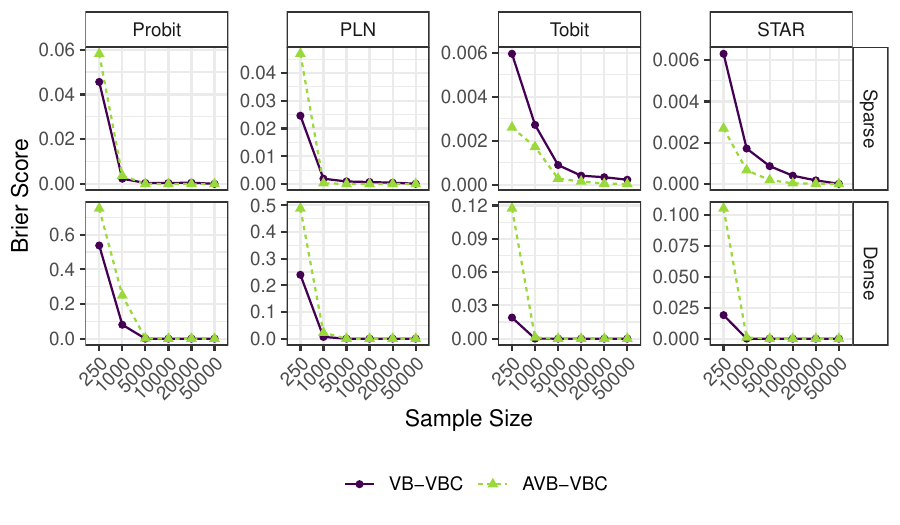}
    \caption{Brier scores for posterior inclusion probability estimation in the $p=30$ setting using model space exploration algorithm (Section~\ref{sec:algorithmic}). Lower values indicate better performance. Comparison of VBC criterion with full VB (VB-VBC) versus AVB (AVB-VBC). Results based on 10,000 models visited after 2,000 burn-in iterations, averaged across 40 replicate datasets. Sample sizes range from $n=250$ to $n=50,000$.}
    \label{fig:brier_exploration}
\end{figure}

\begin{figure}[ht]
    \centering
    \includegraphics[width=\textwidth]{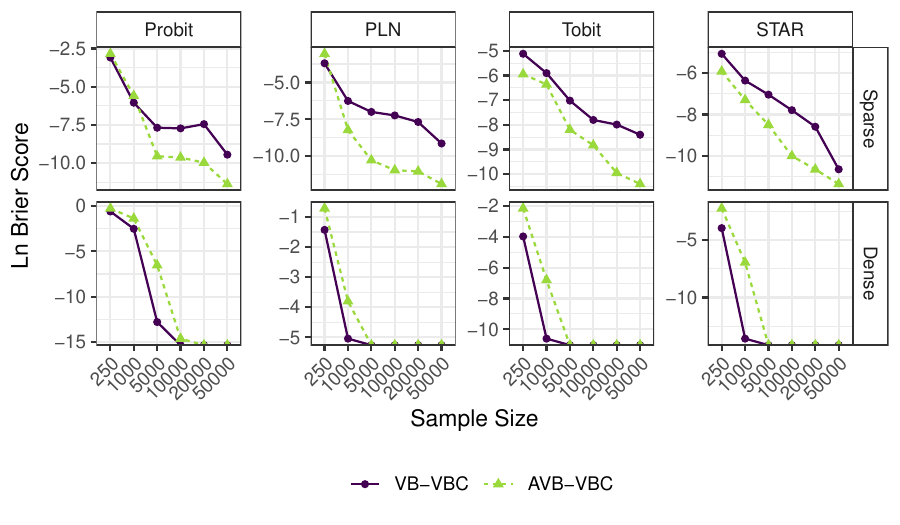}
    \caption{Log scale Brier scores for posterior inclusion probability estimation in the $p=30$ setting using model space exploration algorithm (Section~\ref{sec:algorithmic}). Lower values indicate better performance. Comparison of VBC criterion with full VB (VB-VBC) versus AVB (AVB-VBC). Results based on 10,000 models visited after 2,000 burn-in iterations, averaged across 40 replicate datasets. Sample sizes range from $n=250$ to $n=50,000$.}
    \label{fig:brier_exploration_log}
\end{figure}

\begin{figure}[ht]
    \centering
    \includegraphics[width=\textwidth]{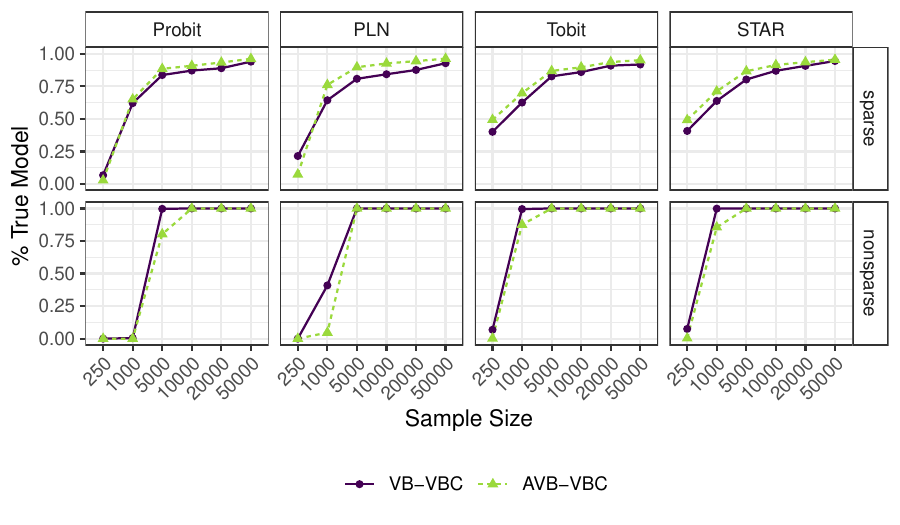}
    \caption{Proportion of iterations in which the true model was visited in the $p=30$ setting using model space exploration algorithm. Comparison of VBC criterion with full VB (VB-VBC) versus AVB (AVB-VBC). Results based on 10,000 models visited after 2,000 burn-in iterations, averaged across 40 replicate datasets. Sample sizes range from $n=250$ to $n=50,000$.}
    \label{fig:true_exploration}
\end{figure}

\begin{figure}[ht]
    \centering
    \includegraphics[width=\textwidth]{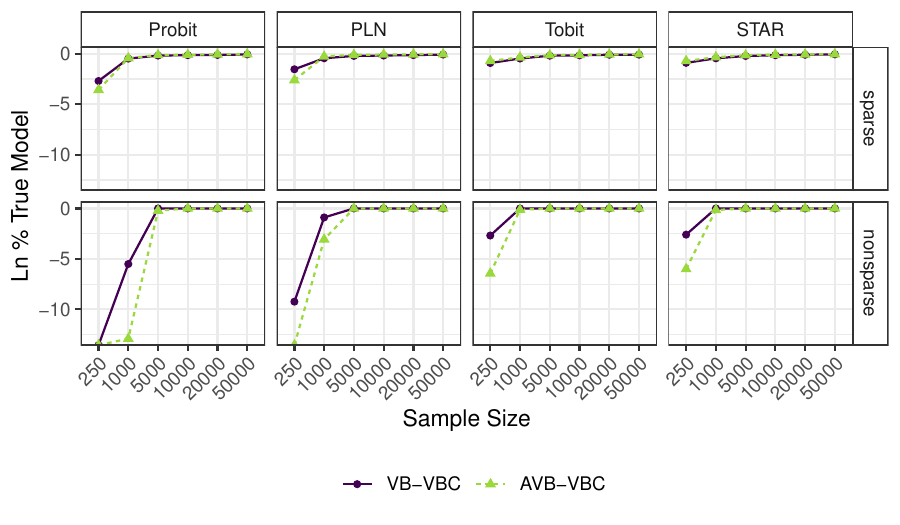}
    \caption{Log scale proportion of iterations in which the true model was visited in the $p=30$ setting using model space exploration algorithm. Comparison of VBC criterion with full VB (VB-VBC) versus AVB (AVB-VBC). Results based on 10,000 models visited after 2,000 burn-in iterations, averaged across 40 replicate datasets. Sample sizes range from $n=250$ to $n=50,000$.}
    \label{fig:true_exploration_log}
\end{figure}

\begin{figure}[ht]
    \centering
    \includegraphics[width=\textwidth]{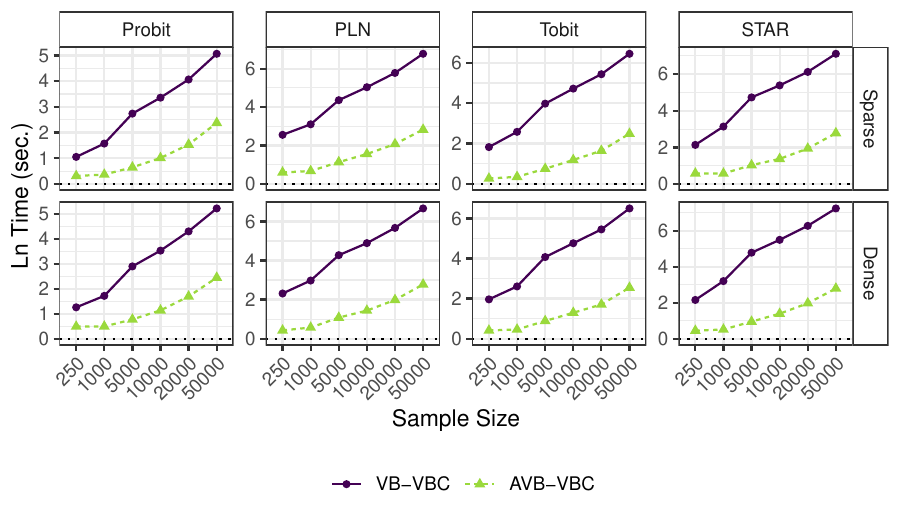}
    \caption{Log scale computation time (in seconds) required to enumerate all 1,024 models in the $p=10$ setting. Comparison of VBC criterion with full VB (VB-VBC) versus AVB (AVB-VBC). Results averaged across 40 replicate datasets. Sample sizes range from $n=250$ to $n=50,000$.}
    \label{fig:timing_enumeration_log}
\end{figure}

\begin{figure}[ht]
    \centering
    \includegraphics[width=\textwidth]{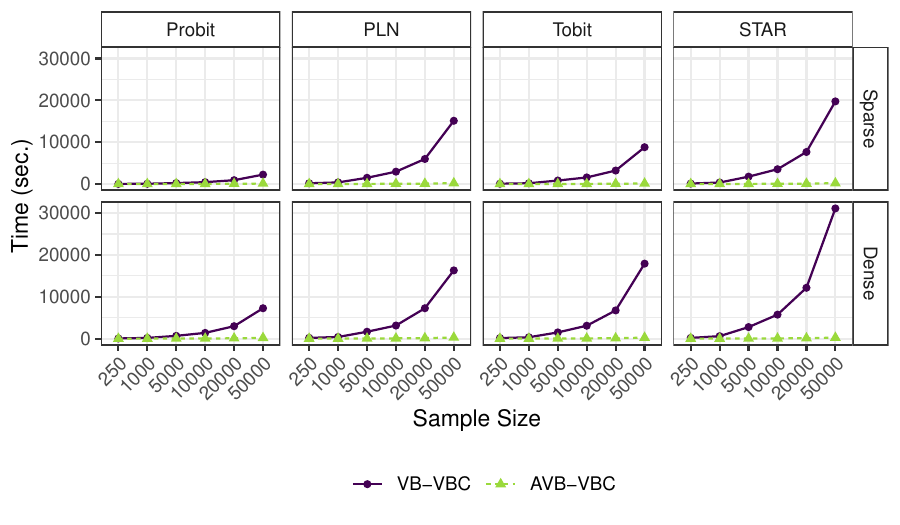}
    \caption{Computation time (in seconds) required to explore 10,000 models in the $p=30$ setting using the model space exploration algorithm. Comparison of VBC criterion with full VB (VB-VBC) versus AVB (AVB-VBC). Results include 2,000 burn-in iterations and are averaged across 40 replicate datasets. Sample sizes range from $n=250$ to $n=50,000$.}
    \label{fig:timing_exploration}
\end{figure}

\begin{figure}[ht]
    \centering
    \includegraphics[width=\textwidth]{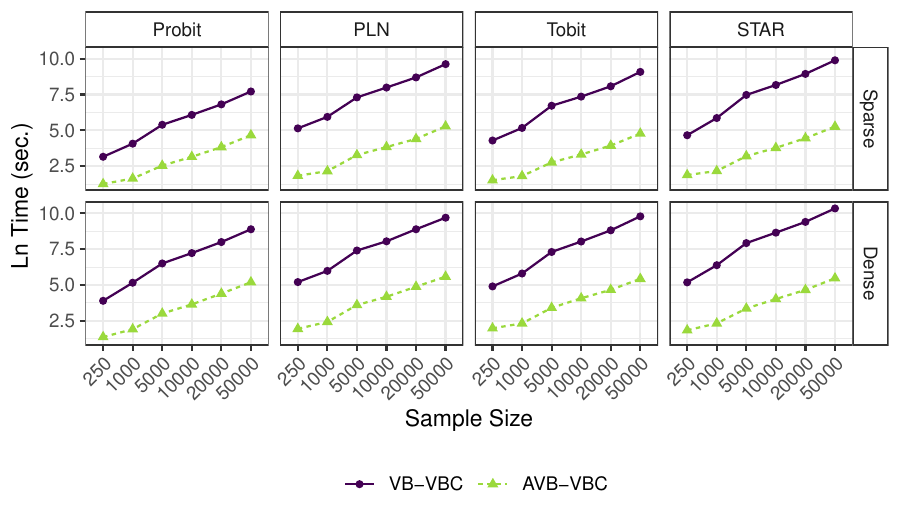}
    \caption{Log scale computation time (in seconds) required to explore 10,000 models in the $p=30$ setting using the model space exploration algorithm. Comparison of VBC criterion with full VB (VB-VBC) versus AVB (AVB-VBC). Results include 2,000 burn-in iterations and are averaged across 40 replicate datasets. Sample sizes range from $n=250$ to $n=50,000$.}
    \label{fig:timing_exploration_log}
\end{figure}

\begin{figure}[ht]
    \centering
    \includegraphics[width=\textwidth]{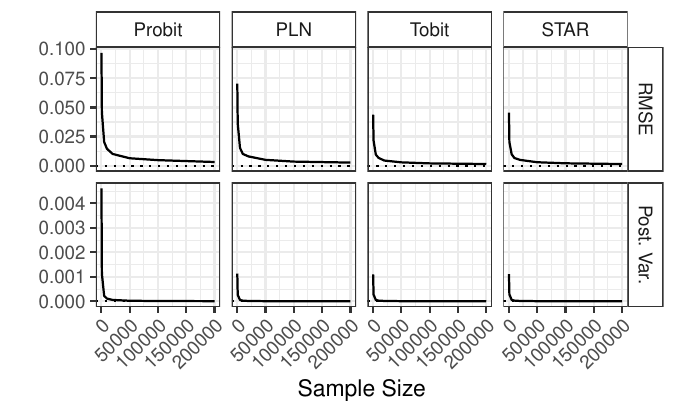}
    \caption{Empirical assessment of posterior consistency for full VB algorithms. Top row: Root mean square errors (RMSE) of parameter estimates relative to true values for $\alpha$, $\bm{\beta}$, and $\sigma^2$ (when applicable). Bottom row: Average posterior variance estimates for the same parameters. Results based on full VB algorithm across four models (Probit, Tobit, STAR, PLN) with sample sizes from $n=250$ to $n=200,000$. All values averaged across 30 replicate datasets. Both RMSE and posterior variances approach zero as $n$ increases, providing empirical evidence of posterior consistency.}
    \label{fig:consistency}
\end{figure}

\begin{figure}[ht]
    \centering
    \includegraphics[width=\textwidth]{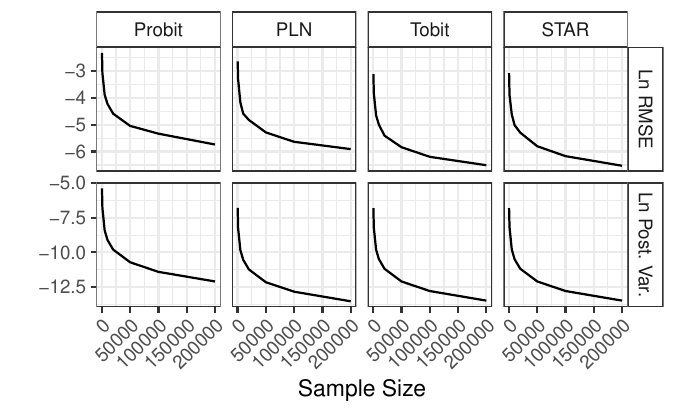}
    \caption{Log scale values of RMSE and posterior variances for empirical assessment of posterior consistency for full VB algorithms. Top row: Root mean square errors (RMSE) of parameter estimates relative to true values for $\alpha$, $\bm{\beta}$, and $\sigma^2$ (when applicable). Bottom row: Average posterior variance estimates for the same parameters. Results based on full VB algorithm across four models (Probit, Tobit, STAR, PLN) with sample sizes from $n=250$ to $n=200,000$. All values averaged across 30 replicate datasets.}
    \label{fig:consistency_log}
\end{figure}

\end{document}